\documentclass[a4paper,UKenglish, cleveref, autoref, thm-restate, numberwithinsect]{lipics-v2021}

\usepackage{xspace}
\usepackage{enumitem}

\def\FULL{true}
\ifdefined\FULL{}
\newcommand{\full}[1]{#1}
\newcommand{\short}[1]{}
\newcommand{\starr}{$\star$}
\else
\newcommand{\full}[1]{}
\newcommand{\short}[1]{#1}
\newcommand{\starr}{$\star$}
\fi

\newenvironment{innerproof}
 {\proof}
 {\endproof}

\newcommand{\defproblem}[3]{
	\vspace{2mm}
	\noindent\fbox{
		\begin{minipage}{0.96\linewidth}
			\begin{tabular*}{\linewidth}{@{\extracolsep{\fill}}lr} \textsc{#1} &  \\ \end{tabular*}
			{\bf{Input:}} #2 \\
			{\bf{Question:}} #3
		\end{minipage}
	}
	\vspace{1mm}
}

\newcommand{\gcal}{\mathcal{G}}
\newcommand{\bcal}{\mathcal{H}}
\newcommand{\scal}{\mathcal{S}}
\newcommand{\tw}{\mathrm{\textbf{tw}}}
\newcommand{\ttw}{\mathbb{T}}
\newcommand{\td}{\mathrm{\textbf{td}}}

\newcommand{\comp}{\texttt{Comp}}
\newcommand{\compress}{\texttt{Condense}}
\newcommand{\aux}{\texttt{Aux}}
\newcommand{\minsep}{\texttt{MinSep}}
\newcommand{\base}{\texttt{Base}}
\newcommand{\sign}{\texttt{Sign}}
\newcommand{\spann}{\texttt{Span}}
\newcommand{\minver}{minimal vertex separator\xspace}

\newcommand{\br}[1]{\left(#1\right)}

\newcommand{\eps}{\varepsilon}
\newcommand{\sm}{\setminus}
\newcommand{\sub}{\subseteq}

\newcommand{\Oh}{\mathcal{O}}
\newcommand{\hh}{\ensuremath{\mathcal{H}}}

\ifdefined\DEBUG{}

\def\rem#1{{\marginpar{\raggedright\scriptsize #1}}}
\newcommand{\micr}[1]{\rem{\textcolor{blue}{\(\bullet \) #1}}}

\newcommand{\bmpr}[1]{\rem{\textcolor{purple}{\(\bullet \) #1}}}

\newcommand{\jjhr}[1]{\rem{\textcolor{orange}{\(\bullet \) #1}}}

\else

\newcommand{\micr}[1]{ }
\newcommand{\bmpr}[1]{ }
\newcommand{\jjhr}[1]{ }
\fi

\hideLIPIcs
\title{Tight Bounds for Chordal/Interval Vertex Deletion Parameterized by Treewidth}
\titlerunning{Tight bounds for Chordal/Interval Vertex Deletion} 
\author{Micha\l{} W\l{}odarczyk}{Ben-Gurion University, Beer Sheva, Israel}{michal.wloda@gmail.com}{https://orcid.org/0000-0003-0968-8414}{}
\authorrunning{Micha\l{} W\l{}odarczyk}
\funding{Supported by the European Research Council (ERC) grant titled PARAPATH. Part of the work has been carried out when the author was a postdoctorate student at the Eindhoven University of Technology.}
\Copyright{Micha\l{} W\l{}odarczyk}



\category{} 

\relatedversion{} 

\nolinenumbers 

\EventEditors{Kousha Etessami, Uriel Feige, and Gabriele Puppis}
\EventNoEds{3}
\EventLongTitle{50th International Colloquium on Automata, Languages, and Programming (ICALP 2023)}
\EventShortTitle{ICALP 2023}
\EventAcronym{ICALP}
\EventYear{2023}
\EventDate{July 10--14, 2023}
\EventLocation{Paderborn, Germany}
\EventLogo{}
\SeriesVolume{261}
\ArticleNo{3}

\begin{document}

\maketitle

\begin{abstract}
In \textsc{Chordal/Interval Vertex Deletion} we ask how many vertices one needs to remove from a~graph to make it chordal (respectively: interval).
We study these problems under the parameteri\-zation by treewidth $\tw$ of the input graph $G$.
On the one hand, we present an algorithm for \textsc{Chordal Vertex Deletion} with running time 
$2^{\Oh(\tw)} \cdot |V(G)|$,
improving upon the running time $2^{\Oh(\tw^2)} \cdot |V(G)|^{\Oh(1)}$ by Jansen, de Kroon, and W\l{}odarczyk (STOC'21).
When a tree decomposition of width $\tw$ is given, then the base of the exponent equals $2^{\omega-1}\cdot 3 + 1$.
Our algorithm is based on a~novel link between chordal graphs and graphic matroids, which allows us to employ the framework of repre\-sentative families.
On~the other hand, we prove that the known  \mbox{$2^{\Oh(\tw \log \tw)} \cdot |V(G)|$-time} algorithm for \textsc{Interval Vertex Deletion}  cannot be improved assuming Exponential Time~Hypothesis.
\end{abstract}

\section{Introduction}\label{sec:intro}

Treewidth \cite[\S 7]{blue-book} is arguably the most extensively studied width measure in the graph theory.
Simply speaking, treewidth measures to what extent a graph is similar to a tree, where trees and forests are exactly the graphs of treewidth 1.
It plays a crucial role in Robertson and Seymour's Graph Minors series~\cite{GM2}.
The usefulness of treewidth stems from the fact that a broad class of problems can be solved in linear time on graphs of bounded treewidth.
The celebrated Courcelle's Theorem~\cite{CourcelleE12} states that any graph problem expressible in the Counting Monadic Second Order Logic ({CMSO}) can be solved in time $f(\tw)\cdot |V(G)|$, where $\tw$ denotes the treewidth of graph $G$ and $f$ is some computable function.
In other words, every such problem is fixed-parameter tractable (FPT) when parameterized by treewidth.
Furthermore, bounded-treewidth graphs appear in a wide variety of contexts, which makes treewidth-based algorithms a~ubiquitous tool in algorithm design~\cite{DemaineH08, JansenLS14, Marx20, MarxOR13, RobertsonS95}. 

The function $f$ from Courcelle's Theorem may grow very rapidly and a large body of research has been devoted to optimize the dependency on $\tw$ for particular problems.
In the ideal scenario, we would like the function $f$ to be single-exponential, i.e., $f(\tw) = 2^{\Oh(\tw)}$, while possibly 
allowing a higher (yet constant) exponent at $|V(G)|$.
This is often the best we can hope for
because sub-exponential running times usually
contradict the Exponential Time Hypothesis\footnote{The 
Exponential Time Hypothesis states that there exists a constant $\delta > 0$ so that 3-SAT cannot be solved in time $\Oh(2^{\delta n})$ on $n$-variable formulas.
} (ETH)~\cite{eth}.

While the standard dynamic programming technique yields single-exponential algorithms for problems with `local constraints', such as \textsc{Vertex Cover}, \textsc{Dominating Set}, or \textsc{Bipartization}, it falls short for problems with `connectivity constraints', such as \textsc{Feedback Vertex Set}, \textsc{Hamiltonian Cycle}, or \textsc{Connected Vertex Cover}, leading to parameter dependency $f(\tw) = 2^{\Oh(\tw \log \tw)}$.
On the one hand, this issue was dealt with in the landmark work of Cygan et al.~\cite{cygan2011solving}, who introduced the Cut \& Count technique and obtained randomized single-exponential algorithms for the problems above, among others (see also~\cite{pilipczuk-single-expo-logic}).
In following works, Bodlaender et al.~\cite{BodlaenderCKN15} and Fomin et al.~\cite{representative-efficient} presented alternative techniques that allow to circumvent randomization: matrix-based approaches and representative families.
On the other hand, Lokshtanov et al.~\cite{LokshtanovMS18}
provided a framework for proving `slightly super-exponential' lower bounds under ETH, which paved the way for establishing tight lower bounds for problems that require dependency $f(\tw) = 2^{\Oh(\tw \log \tw)}$.
In the same work, they obtained such bounds for \textsc{Disjoint Paths} and \textsc{Chromatic Number}.
For problems with a~single-exponential dependency $f(\tw) = \Oh(c^\tw)$, further research has been devoted to establish the optimal base of the exponent~$c$~\cite{Radu18, CyganKN18, cygan2011solving, LokshtanovMS18optimal, van2009dynamic}.

\subparagraph{Vertex-deletion problems.}
Many optimization graph problems can be phrased in terms of \mbox{\textsc{$\hh$-Vertex Deletion}}: remove the smallest number of vertices from a graph so that the resulting graph belongs to the graph class $\hh$. For example, \textsc{Vertex Cover} corresponds to the class $\hh$ of edge-less graphs. 
There is a diverse complexity landscape of ETH-tight running times for various vertex-deletion problems under treewidth parameterization.
The classes $\hh$ for which tight bounds have been established include:  edge-less graphs~\cite{LokshtanovMS18optimal},  forests~\cite{cygan2011solving} (see also~\cite{closeRelatives}), planar graphs~\cite{JansenLS14, Pilipczuk17},
classes defined by a connected forbidden minor~\cite{BasteST20} (see also~\cite{BasteST20j1, BasteST20j2, BasteST20j3}), bipartite graphs~\cite{LokshtanovMS18optimal}, DAGs~\cite{BonamyKNPSW18},
even-cycle-free graphs~\cite{closeRelatives, closeRelativesRevisit}, and some classes defined by a forbidden (induced) subgraph~\cite{CyganHitting, SauHitting}.

We extend this list by studying the vertex-deletion problems into the classes of chordal and interval graphs.
A graph is chordal if it does not contain an induced cycle of length at least 4 (a hole) and a graph is interval if it is an intersection graph of intervals on the real line.
Any interval graph is chordal and any chordal graph is perfect.
Applications of these two graph classes have been  long studied in miscellaneous areas of discrete optimization~\cite{BarNoy01, Benzer59, Buneman74, kendall1969incidence, Prim57, rose1972graph}.
On the theoretical side, the treewidth (resp. pathwidth) of a graph $G$ equals the minimum clique number of a~chordal (resp. interval) supergraph of $G$~\cite{blue-book, pathwidthInterval}.
Moreover, some hard problems become tractable on chordal or interval graphs (or even on graphs with small vertex-deletion distance to chordality)~\cite{Cai03, JacobP0S20, Keil85}.


\subparagraph{Our results.}
The state of the art for \textsc{Chordal Vertex Deletion} (\textsc{ChVD}) is the running time $2^{\Oh(\tw^2)}n^{\Oh(1)}$, which follows from a more general result for a hybrid graph measure $\hh$-treewidth, where $\hh = \texttt{chordal}$~\cite{JansenKW21}.
We~improve the dependency on treewidth to single-exponential.

\begin{restatable}{theorem}{thmChordal}
\label{thm:intro:chordal}
\textsc{Chordal Vertex Deletion} can be solved in deterministic time $\Oh(c^k k^{\omega+1} n)$ on $n$-vertex node-weighted graphs when a tree decomposition of width $k$ is provided.
The constant $c$ equals $2^{\omega-1} \cdot 3 + 1$.
\end{restatable}

Here, $\omega < 2.373$ stands for the matrix multiplication exponent~\cite{omega}.
To prove \cref{thm:intro:chordal} we establish a~new link between chordal graphs and graphic matroids, which allows us to exploit the framework of representative families~\cite{representative-product, representative-efficient}.
\textsc{ChVD}  is at least as hard as \textsc{Feedback Vertex Set}, what implies barriers for a significant improvement in the constant~$c$ (see \cref{lem:chordal:LB} and the discussion therein). 
Thanks to a single-exponential constant-factor FPT approximation for treewidth~\cite{BodlaenderDDFLP16}, \cref{thm:intro:chordal} gives running time $2^{\Oh(\tw)}n$ even when no tree decomposition is provided in the input.

The best known running time for \textsc{Interval Vertex Deletion} is $2^{\Oh(\tw \log \tw)}n$~\cite{intervalDP}.
(While this algorithm has been described for the edge-deletion variant, we briefly explain
\full{in \cref{app:intervalDP}}
\short{in the full version of the article}
how it can be adapted for vertex deletion.)
We show that, unlike the chordal case, this running time is optimal under ETH.
This gives a sharp separation between the two studied problems. 

\begin{theorem}\label{thm:intro:interval}
Under the assumption of \textsc{ETH},
\textsc{Interval Vertex Deletion} cannot be solved in time $2^{o(\tw \log \tw)}n^{\Oh(1)}$ on $n$-vertex unweighted graphs of treewidth $\tw$.
\end{theorem}

In fact, we show a stronger lower bound that rules out the same running time with respect to a different graph parameter, called treedepth, which is never smaller than treewidth.
Our lower bound is obtained via a reduction from $k\times k$ \textsc{Permutation Clique}~\cite{LokshtanovMS18}, which produces an instance of size $2^{\Oh(k)}$ and treedepth $\Oh(k)$.

\subparagraph{Related work.}

The two considered $\hh$-\textsc{Vertex Deletion} problems 
have been studied in several contexts.
Both problems are FPT parameterized by the solution size $k$, with the best-known running times $\Oh(8^k(n+m))$ for $\hh = \texttt{interval}$~\cite{CaoK16}
and $2^{\Oh(k \log k)}n^{\Oh(1)}$ for $\hh = \texttt{chordal}$~\cite{CaoM15} (but the problem becomes W[2]-hard for $\hh = \texttt{perfect}$~\cite{Heggernes13}). 
There are polynomial-time approximation algorithms with approximation factor~8 for $\hh = \texttt{interval}$~\cite{CaoK16} and $k^{\Oh(1)}$ for $\hh = \texttt{chordal}$~\cite{JansenP18}.
Observe that, in these two regimes, vertex deletion into chordal graphs seems harder than into interval graphs (although no lower bounds are known to justify such a separation formally); this contrasts our results with respect to the treewidth parameterization. 

Both studied problems admit exact exponential algorithms with running times of the form $\Oh((2-\eps)^n)$~\cite{bliznets2016largest} as well as polynomial kernelizations~\cite{AgrawalLMSZ19, AgrawalPSZ19, JansenP18}.
The obstructions to being chordal (resp. interval) enjoy the Erd\H{o}s-P\'{o}sa property: any graph $G$ either contains $k$ vertex-disjoint subgraphs which are not chordal (resp. not interval) or a~vertex set $X$ of size $\Oh(k^2 \log k)$ such that $G-X$ is chordal~\cite{KimK20} (resp. interval~\cite{AgrawalLM0Z18}).
Vertex deletion into other subclasses of perfect graphs has been studied as well~\cite{AgrawalKLS16, AhnEKO20, AhnKL22, HofV13}.
For other modification variants, where instead of vertex deletions one considers removals, insertions, or contractions of edges, see, e.g.,~\cite{BliznetsCKPM20, Cai03, CaoK16, CaoM15, FominV13, LokshtanovMS13, Yannakakis81}.

The concept of representative families, which plays an important role in our algorithm for \textsc{ChVD}, has found applications outside the context of treewidth as well~\cite{zehavi2016partial, zehavi2015mixing}.
Our other tool, boundaried graphs, has revealed fruitful insights
for various graph classes~\cite{BasteST20, meta, JansenKW21}.

\subparagraph{Organization of the paper.}

We begin by describing our technical contributions informally in \cref{sec:techniques}.
\short{We provide basic preliminaries in \cref{sec:prelim}, while the extended preliminaries including tree decompositions and representative families can be found in the full version of the article.
\cref{sec:chordal} is devoted to establishing a connection between chordal graphs and graphic matroids.
The description of the dynamic programming algorithm over a tree decomposition follows standard conventions and is provided in the full version.}
\full{In \cref{sec:prelim} we provide formal preliminaries.
\cref{sec:chordal} is devoted to establishing a connection between chordal graphs and graphic matroids, which is followed by the proof of \cref{thm:intro:chordal}.}
In \cref{sec:interval} we prove our lower bound for \textsc{Interval Vertex Deletion}.
We conclude in \cref{sec:conclusion}.
\short{The proofs of statements indicated with $(\star)$ are postponed to the full version.
}

\section{Techniques}
\label{sec:techniques}

\subparagraph{Chordal Vertex Deletion.}
The standard approach to design algorithms over a bounded-width tree decomposition is to assign a data structure to each node $t$ in the decomposition, which stores information about partial solutions for the subgraph associated with the subtree of $t$.
Suppose that $X \sub V(G)$ is a bag of $t$, $A \sub V(G) \sm X$ denote the set of vertices appearing in the bags of the descendants of $t$ (but not in $X$), and $B \sub V(G)$ is the set of remaining vertices.
We say that a subset $S \sub V(G)$ is a \emph{solution} if $G[S]$ is chordal; we want to \emph{maximize} the size of $S$.
Next, a pair $(S_A \sub A, S_X \sub X)$ is a \emph{partial solution} if $G[S_A \cup S_X]$ is chordal.
A set $S_B \sub B$ is an \emph{extension} of a partial solution $(S_A, S_X)$ if $S_A \cup S_X \cup S_B$ is a solution.
Since $X$ separates $S_A$ from $S_B$, the graph $G[S_A \cup S_X \cup S_B]$ can be regarded as a result of \emph{gluing}  $G[S_A \cup S_X]$ with  $G[S_B \cup S_X]$ alongside the \emph{boundary} $S_X$.
For a node $t$ and $S_X \subseteq X$, we want to store a family of partial solutions $\gcal_{t,S_X}$ so that for every possible  $S_B \sub B$: if $S_B$ is an extension for some partial solution $(S_A,S_X)$, then there exists a partial solution $(S'_A,S_X)\in\gcal_{t,S_X}$ for which (a) $S_B$ is still a valid extension, and 
(b) $S'_A$ is at least as large as $S_A$.
We say that such a family satisfies the correctness invariant for $(t,S_X)$.

Jansen et al.~\cite{JansenKW21} showed that any chordal graph $H$ with a boundary of size $k$ can be \emph{condensed} to a graph $H'$ on $\Oh(k)$ vertices that exhibits the same behavior in terms of gluing.
More precisely, the gluing product of $H$ with any graph $J$ is chordal if and only if the gluing product of $H'$ with $J$ is chordal.
Since there are $2^{\Oh(\tw^2)}$ graphs on $\Oh(\tw)$ vertices and $2^{\Oh(\tw)}$ choices for the boundary $S_X$, it suffices to store only $2^{\Oh(\tw^2)}$ partial solutions.

We take this idea one step further and show that it is actually sufficient to store only $2^{\Oh(\tw)}$ partial solutions.
To this end, we investigate the properties of the class of chordal graphs with respect to the gluing operation and prove a homomorphism theorem relating it to graphic matroids. 
A \emph{graphic matroid} of a graph $J$ is a set system $\mathcal{I}$ over $E(J)$ where a subset $S \sub E(J)$
belongs to  $\mathcal{I}$ (and is called \emph{independent}) when $S$ contains no cycles.
A~\emph{rank} of a matroid is the largest size of an independent set; here this coincides with the size of any spanning forest in $J$.
In the following statement, $\mathcal{G}_{X,B}$ is a family of graphs  $H$ that
satisfy (a) $V(H) \supseteq X$ and (b) $H[X] = B$.
For graphs $H_1, H_2 \in \mathcal{G}_{X,B}$ we
assume that $V(H_1) \cap V(H_2) = X$ and
define their gluing product as $H_3 = (H_1,X) \oplus (H_2,X)$ where $V(H_3) = V(H_1) \cup V(H_2)$ and $E(H_3) = E(H_1) \cup E(H_2)$.

\begin{theorem}\label{prop:intro:homo}
Consider a family of graphs $\mathcal{G}_{X,B}$ for some pair $(X,B)$. 
There exists a~graphic matroid $M=(E,\mathcal{I})$ of rank at most $|X|-1$ and a polynomial-time computable mapping $\sigma: \mathcal{G}_{X,B} \to 2^E$ such that $(H_1,X) \oplus (H_2,X)$ is chordal if and only if \mbox{$\sigma(H_1) \cap \sigma(H_2) = \emptyset$} and $\sigma(H_1) \cup \sigma(H_2) \in \mathcal{I}$.
\end{theorem}

With this criterion at hand, we can employ the machinery of representative families to truncate the number of partial solutions to be stored for a node of a tree decomposition.
Technical details aside, for a family $\scal$ of independent sets in a matroid $M=(E,\mathcal{I})$, a~subfamily $\widehat \scal \sub \scal$ is called \emph{representative} for $\scal$ if for every independent set $Y$ in $M$: if there exists $X \in \scal$ so that $X \cap Y = \emptyset$ and $X \cup Y \in \mathcal{I}$, then there exists $\widehat X \in \widehat \scal$ so that $\widehat X \cap Y = \emptyset$ and $\widehat X \cup Y \in \mathcal{I}$.
Fomin et al.~\cite{representative-efficient} showed that for any family $\scal$ in a graphic matroid (more generally, in a linear matroid) of rank $k$ there exists a representative family of size at most $2^k$ and it can be constructed in time $2^{\Oh(k)}$.
We~use \cref{prop:intro:homo} to translate this result into the language of chordal graphs and gluing.
When $\mathcal{G}_{t,S_X}$ is a family of 
partial solutions that satisfies the correctness invariant for $(t,S_X)$, a~representative family for $\sigma(\gcal_{t,S_X})$ in the related graphic matroid $M$ corresponds to a subfamily  $\widehat  {\mathcal{G}}_{t,S_X} \sub \gcal_{t,S_X}$ that satisfies condition~(a) of the correctness invariant and $|\widehat  {\mathcal{G}}_{t,S_X}| \le 2^{\tw}$.
In order to satisfy condition (b), we need to assign weights to the elements of the matroid $M$, encoding the size of the largest partial solution mapped to each element.
We can then utilize the weighted variant of representative families, which preserves the largest-weight elements~\cite{representative-efficient}.
By storing only the condensed forms of the partial solutions (having $\Oh(\tw)$ vertices), we also achieve a linear dependency on $|V(G)|$.

In order to prove \cref{prop:intro:homo}, 
we give a novel criterion for testing chordality of a gluing product.
When $G$ originates from gluing two chordal graphs $G_1,G_2$ alongside boundary $X$, then any hole in $G$ must visit both $V(G_1)\sm X$ and  $V(G_2)\sm X$, 
so it must traverse $X$ multiple times.
We show that if a hole $H$ intersects at least two connected components of $G[X]$, then it corresponds to a cycle in the graph obtained from $G$ by contracting each of the connected components of $G[X]$, $G_1-X$, $G_2-X$ into single vertices.
Otherwise, let $C$ be the unique connected component of $G[X]$ that is intersected by the hole.
We prove that
there exists a vertex set $S \sub V(C)$ that is disjoint from $V(H)$ and $C-S$ has two connected components $C_1, C_2$ satisfying $N_C(C_1) = N_C(C_2) = S$ (below we refer to such components as \emph{relevant}) and
having non-empty intersections with $V(H)$.
Moreover, every vertex from $V(H) \cap C$ belongs to some relevant component.
Consider a graph $\aux(G,X,S)$ obtained from $G$ by (1) removing the connected components of $G[X]$ different than $C$, (2) contracting relevant components of $C-S$ into single vertices while removing the irrelevant ones, and (3)~contracting the components of $G_1-X$, $G_2-X$ into single vertices.
A detailed construction is given in \cref{def:chordal:aux}; see also Figure \ref{fig:aux} on page \pageref{fig:aux}.
Then the hole $H$ corresponds to a cycle in  $\aux(G,X,S)$.
The first scenario can be analyzed with this approach as well, by taking $S=\emptyset$.
We prove that considering all \minver{s} $S$ in $G[X]$ and checking acyclity of each auxiliary graph $\aux(G,X,S)$ yields a necessary and sufficient condition for $G$ to be chordal. 

This criterion allows us to construct a graphic matroid encoding all the information about 
each of the graphs $G_1, G_2$ necessary to reconstruct the graphs $\aux(G,X,S)$ and to determine whether $G$ is chordal.
In~order to bound the rank of this matroid, we investigate the structure of \minver{s} in a chordal graph and bound the 
size of a spanning forest in a certain graph obtained
from the union of $\aux(G,X,S)$.
A~criterion of a~similar kind is known for testing planarity of a~gluing product of planar graphs
when the boundary has a~Hamiltonian cycle; then the corresponding auxiliary graph (defined in a different way) should be bipartite~\cite{DiETT99}. 
\short{Our criterion can be also compared to the one used by Bonnet et al.~\cite{BonnectBKM17} for analyzing gluing products with respect to certain subclasses of chordal graphs.
We elaborate more on their approach in the full version of the paper.}

\full{Our criterion can be also compared to the one used by Bonnet et al.~\cite{BonnectBKM17} in their work on \textsc{Bounded $\mathcal{P}$-Block Vertex Deletion}.
Here, the task is to remove the smallest number of vertices from a graph so that every remaining biconnected component has at most $d$ vertices and belongs to the class  $\mathcal{P}$.
They showed that when $\mathcal{P}$ is a subclass of chordal graphs then \textsc{Bounded $\mathcal{P}$-Block Vertex Deletion} can be solved in time $2^{\Oh(\tw \cdot d^2)} n^{\Oh(1)}$ and otherwise it cannot be solved in time $2^{o(\tw \log \tw)} n^{\Oh(1)}$ for fixed $d$ unless ETH fails.
Their positive result is also based on a criterion which determines whether a gluing product of two graphs has the desired property by checking if a union of two certain sets is independent in a graphic matroid.
It handles cases similar to the first scenario considered in the outline above.
However, while this criterion is necessary it is not 
sufficient and more information needs to be stored in a DP state, leading to the additional factor $d^2$ in the exponent.
In our setting the biconnected components can be arbitrarily large so such a factor is prohibitive.}

\subparagraph{Interval Vertex Deletion.}
In order to prove \cref{thm:intro:interval} we present a parameterized reduction from $k\times k$ \textsc{Permutation Clique}.
Here, the input is a graph $G$ on vertex set $[k] \times [k]$, and we ask whether there exists a permutation $\pi \colon [k] \to [k]$ such that $(1,\pi(1)), (2,\pi(2)), \dots, (k,\pi(k))$ forms a~clique in $G$.
Lokshtanov et al.~\cite{LokshtanovMS18} proved that $k\times k$ \textsc{Permutation Clique} cannot be solved in time $2^{o(k \log k)}$ under ETH.
So we seek a reduction from  $k\times k$ \textsc{Permutation Clique} to \textsc{Interval Vertex Deletion} that produces a graph of treewidth $\Oh(k)$.

Imagine an interval model of a complete graph $Y$ on vertex set $[k]$ in which all the right endpoints of the intervals coincide and all the left endpoints are distinct.
Choosing the order of the left endpoints encodes some permutation $\pi \colon [k] \to [k]$
(see \Cref{fig:permutation} on page \pageref{fig:permutation}).
We can extend this interval model by inserting a new vertex $v$ only if $N(v)$ corresponds to a set of intervals intersecting at a single point.
This is possible only when $N(v) = \pi([\ell])$ for some $\ell \in [k]$.
Furthermore, inserting to $Y$ independent vertices $v_1, v_2, \dots, v_k$, such that $|N(v_i)| = i$ and $N(v_i) \subset N(v_{i+1})$, enforces the choice of permutation~$\pi$.
We can thus encode a permutation $\pi$ by an ascending family of sets $N_1 \subset N_2 \subset \dots \subset N_k = [k]$, satisfying $N_i = \pi([i])$, which correspond to the neighborhoods of $v_1, v_2, \dots, v_k$ in $Y$.
On the other hand, any ascending family of sets for which the construction above gives an interval graph, must encode some permutation.
On an intuitive level, a partial interval model of a size-$k$ separator can encode one of $k!$ permutations.

We need a mechanism to verify that a chosen permutation $\pi$ encodes a clique, i.e., that it satisfies $k \choose 2$ constraints of the form $(i,\pi(i))(j,\pi(j)) \in E(G)$.
To implement a single constraint, we construct a \emph{choice gadget}, inspired by the reduction to 
\textsc{Planar Vertex Deletion}~\cite{Pilipczuk17}.
Such a gadget $C_{i,j}$ is defined as a path-like structure, divided into blocks, so that each block has some special vertices adjacent to $Y$ (see \Cref{fig:selector} on page \pageref{fig:selector}). 
We show that any minimum-size interval deletion set in $C_{i,j}$ must `choose' one block and leave its special vertices untouched while it can remove the remaining special vertices.
We use this gadget to check if a permutation $\pi$ encoded by an ascending family of sets $N_1 \subset N_2 \subset \dots \subset N_k$ satisfies the constraint $(i,\pi(i))(j,\pi(j)) \in E(G)$.
As $\pi(i)$ is the only element in $N_i \sm N_{i-1}$, this information can be extracted from the tuple $(N_{i-1},N_i,N_{j-1},N_j)$.
We create a single block in $C_{i,j}$ for each valid tuple.
Since the number of such tuples is $2^{\Oh(k)}$, we need a choice gadget of exponential length, unlike the mentioned reduction which works in polynomial time.
However, producing an~instance of size  $2^{\Oh(k)}$ and treewidth $\Oh(k)$ is still sufficient to achieve the claimed lower bound.

\section{Preliminaries}
\label{sec:prelim}

We write $[k] = \{1,2,\dots, k\}$ and assume that $[0] = \emptyset$.
We abbreviate $X \sm v = X \sm \{v\}$.
For a function $w \colon X \to \mathbb{N}$ and $S \sub X$ we use shorthand $w(S) = \sum_{x\in S} w(x)$.
\short{We follow the standard notational conventions for graphs, which are omitted from this extended abstract. }

\full{
\subparagraph{Graphs.}
We consider {finite}, simple, undirected graphs. {We denote} the vertex and edge sets of a graph{~$G$}  by $V(G)$ and $E(G)$, respectively. 
For a set of vertices $S \subseteq V(G)$, by $G[S]$ we denote the graph induced by $S$. 
We use shorthand $G-v$ and $G-S$ for $G[V(G) \sm v]$ and $G[V(G) \sm S]$, respectively. The open neighborhood $N_G(v)$ of $v \in V(G)$ is defined as $\{u \in V(G) \mid \{u,v\} \in E(G)\}$. The closed neighborhood of~$v$ {is} $N_G[v] = N_G(v) \cup \{v\}$. For $S \subseteq V(G)$, we have $N_G[S] = \bigcup_{v \in S} N_G[v]$ and $N_G(S) = N_G[S] \setminus S$.
When $C$ is a subgraph of $G$ we abbreviate $G[C] = G[V(C)]$ and $N_G(C) = N_G(V(C))$.

For sets $S_1,S_2 \subseteq V(G)$ we denote by $E_G(S_1,S_2)$ the set of edges with one endpoint in $S_1$ and one in $S_2$.
We say that $S_1, S_2$ are adjacent in $G$ if $E_G(S_1,S_2) \ne \emptyset$.
A forest is a graph without cycles.
A set $S \sub V(G)$ is called a feedback vertex set if $G-S$ is a forest.
A clique in a graph $G$ is a vertex set $S$ such that
for each distinct $u,v \in S$ the edge $uv$ belongs to~$E(G)$.

A contraction of $uv \in E(G)$ introduces a new vertex adjacent to all of {$N_G(\{u,v\})$}, after which $u$ and $v$ are deleted. 
For $S \subseteq V(G)$ such that $G[S]$ is connected, we say we contract $S$ if we simultaneously contract all edges in $G[S]$ and introduce a single new vertex adjacent to~$N_G(S)$.}

\subparagraph{Separators.}
For vertices $u,v \in V(G)$ a vertex set $S \sub V(G) \sm \{u,v\}$ is called a $(u,v)$-separator if $u,v$ belong to different connected components of $G-S$.
A $(u,v)$-separator is minimal when no proper subset of it is a $(u,v)$-separator.
A vertex set $S$ is called a minimal vertex separator if $S$ is a minimal $(u,v)$-separator for some  $u,v \in V(G)$.

\begin{lemma}[\starr]
\label{lem:prelim:separator-minimal-comp}
Let $u, v$ be vertices in a graph $G$ and $S$ be a $(u,v)$-separator in $G$.
Denote by $C_u, C_v$ the connected components of $G-S$ that contain respectively $u$ and $v$.
Then $S$ is minimal if and only if $N_{G}(C_u) = N_{G}(C_v) = S$.
\end{lemma}
\full{
\begin{proof}
To see the first implication
suppose w.l.o.g. that $N_G(C_u) \subsetneq S$.
Let $w \in S \setminus N_G(C_u)$.
The connected component of $u$ in the graph $G - (S \sm w)$ is $C_u$ because $N_G(C_u) \subseteq S \setminus  w$.
Therefore $S \setminus w$ is also a $(u,v)$-separator contradicting minimality of $S$.

To see the opposite implication suppose that $S' \subsetneq S$ is also a $(u,v)$-separator.
Let $w \in S \sm S'$.
But $w \in N_G(C_u) \cap N_G(C_v)$ so $u,v$ belong to the same connected component of~$G-S'$.
\end{proof}
}

A vertex (or a vertex set) is called {\em simplicial} if its open neighborhood is a clique. 

\begin{lemma}[\starr]
\label{lem:prelim:separator-minimal-simplicial}
Let $S$ be a minimal vertex separator in a graph $G$. Then $S$ does not contain any simplicial vertices.
\end{lemma}
\full{
\begin{proof}
Suppose that $S$ contains a simplicial vertex $v$.
Next, suppose there are two distinct connected components $C_1, C_2$ of $G-S$ which are adjacent to $v$.
Let $w_1 \in N_G(v) \cap C_1, w_2 \in N_G(v) \cap C_2$.
But then $w_1w_2 \in E(G)$ which contradicts the assumption that $C_1, C_2$ are distinct.
Therefore there is at most one connected component of $G-S$ adjacent to $v$.
But then $S \setminus v$ separates the same pairs of vertices as $S$ does.
This means that $S$ is not a minimal vertex separator.
\end{proof}
}

\subparagraph{Chordal and interval graphs.}

An interval graph is an intersection graph of intervals on the real line. In an interval model $\mathcal{I}_G = \{I(v) \mid v \in V(G)\}$ of a graph $G$, each vertex $v \in V(G)$ corresponds to a closed interval $I(v)$;
there is an edge between vertices~$u$ and~$v$ if and only if~$I(v) \cap I(u) \neq \emptyset$.

A {\em hole} in a graph is an induced (i.e., chordless) cycle of length at least four.
A graph is chordal when it does not contain any hole. 
An equivalent definition states that a~chordal graph is an intersection graph of a family of subtrees in a tree~\cite{Gavril74}.
This implies that any interval graph is chordal.
For more background on these graph classes see surveys~\cite{blair1993introduction, BrandstadtLS99}.

The characterization of the two classes as intersection graphs of intervals/subtrees leads to the following observation.

\begin{observation}\label{lem:prelim:closed}
The classes of chordal and interval graphs are closed under vertex deletions and edge contractions.
\end{observation}

An {\em asteroidal triple} (AT) is a triple of vertices such that for any two of them there exists a path between them avoiding the closed neighborhood of the third.
Interval graphs cannot contain ATs, which is a consequence of a linear ordering of any interval model.
It turns out that this is the only property that separates the two graph classes.

\begin{lemma}[{\cite{BrandstadtLS99}}]
\label{lem:prelim:interval-AT}
A graph is interval if and only if it is chordal and does not contain~an~AT.
\end{lemma}

We collect two more useful facts about chordal graphs.

\begin{lemma}[{\cite{BrandstadtLS99}}]
\label{lem:prelim:simplicial-exists}
Every non-empty chordal graph contains a simplicial vertex.
\end{lemma}

When a chordal graph contains a cycle then it also contains a triangle.
As a bipartite graph does not have any triangles, we obtain the following.

\begin{observation}\label{lem:prelim:bip-chordal}
If a graph is chordal and bipartite, then it is a forest.
\end{observation}

A vertex set $S$ in graph $G$ is called a {\em chordal deletion set} (resp. {\em interval deletion set}) if $G-S$ is chordal (resp. interval).
The \textsc{Chordal/Interval Vertex Deletion} problem is defined as follows.
We are given a graph $G$, a non-negative weight function $w \colon V(G) \to \mathbb{N}$, an integer $p$,
and we ask whether there exists a chordal (resp. interval) deletion set $S$ in $G$ such that $w(S) \le p$.

\subparagraph{Boundaried graphs.}

For a set $X$ and a graph $B$ on vertex set $X$, we define a family $\mathcal{G}_{X,B}$ of graphs $G$
that satisfy (a) $V(G) \supseteq X$, (b) $G[X] = B$.
For graphs $G_1, G_2 \in \mathcal{G}_{X,B}$ we define their gluing product $(G_1,X) \oplus (G_2,X)$ by taking a disjoint union of $G_1$ and $G_2$ and identifying vertices from $X$.
Note that two vertices from $X$ are adjacent in $G_1$ if and only if they are adjacent in $G_2$.

For $X \sub V(G)$ a pair $(G,X)$ is called a boundaried graph.
We say that two boundaried graphs  $(G_1,X), (G_2,X)$ are compatible if $G_1,G_2 \in \gcal_{X,B}$ for some $B$.
We remark that it is common in the literature to define a boundaried graph as a triple $(G,X,\lambda)$ where $\lambda \colon X \to [|X|]$ is a labeling (cf. \cite{BasteST20, meta}).
Since we do not need to perform gluing of abstract boundaried graphs, but only ones originating from subgraphs of a fixed graph, this simpler definition is sufficient.

As an example, consider a graph $G$ and $X \subseteq V(G)$.
Then for any $A \sub V(G) \sm X$ the graph $G[A \cup X]$ belongs to $\mathcal{G}_{X,G[X]}$.
When $A, B \sub V(G) \sm X$ are disjoint and non-adjacent then $G[A \cup B \cup X]$ is isomorphic to $(G[A \cup X], X) \oplus (G[B \cup X], X)$.

\full{
\subparagraph{Tree decompositions.}

\begin{definition}[Treewidth]\label{def:prelim:tw}
A \emph{tree decomposition} of a graph $G$ is a pair $(\ttw, \chi)$ where $\ttw$ is a tree, and~$\chi \colon V(\ttw) \to 2^{V(G)}$ is a function, such that:
\begin{enumerate}
    \item for each~$v \in V(G)$ the nodes~$\{t \mid v \in \chi(t)\}$ form a {non-empty} connected subtree of~$\ttw$, \label{item:tree:based:connected}
    \item for each edge~$uv \in E(G)$ there is a node~$t \in V(\ttw)$ with~$\{u,v\} \subseteq \chi(t)$. \label{item:tree:based:edge}
\end{enumerate}
{The \emph{width} of $(\ttw, \chi)$ is defined as~$\max_{t \in V(\ttw)} |\chi(t)| - 1$.}
The \emph{treewidth} of a graph $G$ (denoted $\tw(G))$ is the minimal width a tree decomposition of $G$.
\end{definition}

\begin{definition}
A tree decomposition $(\ttw, \chi)$ is called \emph{nice} if $\ttw$ is a rooted tree with a root $r$ where $\chi(r) = \emptyset$, each node has at most two children, and each node is of one of the following types.
\begin{enumerate}[nolistsep]
\item \textbf{Base node:} a leaf $t \ne r$ in $\ttw$ with $\chi(t) = \emptyset$.
\item \textbf{Introduce node:} a node $t$ having one child $t'$ for which $\chi(t) = \chi(t') \cup \{v\}$ for some $v \not\in \chi(t')$.
\item \textbf{Forget node:} a node $t$ having one child $t'$ for which $\chi(t) = \chi(t') \sm v$ for some $v \in \chi(t')$.
\item \textbf{Join node:} a node $t$ having two children $t_1, t_2$ for which $\chi(t) = \chi(t_1) = \chi(t_2)$.
\end{enumerate}
\end{definition}

It is well known that any tree decomposition of $G$ of width $k$ can be transformed in linear time into a nice tree
decomposition of width $k$ and with $\Oh(k\cdot |V(G)|)$ nodes~\cite{Kloks94}. 
When a rooted tree decomposition $(\ttw, \chi)$ of $G$ is clear from the context we denote by $V_t$ the set of vertices occurring in the subtree rooted at $t \in V(\ttw)$ and define $U_t = V_t \sm \chi(t)$.

\begin{definition}[Treedepth]
A treedepth of a graph $G$ (denoted $\td(G))$ is defined recursively as follows.
\begin{equation*}
    \td(G) = \begin{cases}
        0 & \mbox{if $G$ is empty} \\
        1 + \min_{v \in V(G)}(\td(G-v)) & \mbox{if $G$ is non-empty and connected} \\
         \max_{i=1}^d(\td(G_i)) & \mbox{if $G$ is disconnected and $G_1, \dots G_d$ are its components}
    \end{cases}
\end{equation*}
\end{definition}

As a direct consequence of this definition, inserting a vertex into a graph can increase its treedepth by at most one.
It is well known that for every graph $\tw(G) \le \td(G)$.

\subparagraph{Matroids.}
We provide only the basic background related to our applications.
For more information about matroids we refer to the survey~\cite{oxley2006matroid}.

\begin{definition}[Matroid]
A pair $M = (E,\mathcal{I})$ where $E$ is a set and $\mathcal{I} \sub 2^E$ is called a matroid if the following conditions hold.
\begin{itemize}[nolistsep]
    \item If $X \sub Y$ and $Y \in \mathcal{I}$ then also  $X \in \mathcal{I}$.
    \item If $X, Y \in \mathcal{I}$ and $|X| < |Y|$ then there exists $e \in Y \sm X$ such that $X \cup \{e\} \in \mathcal{I}$.
\end{itemize}
We say that a set $X \sub E$ is independent in $M$ when $X \in \mathcal{I}$.
The rank of $M$ is the size of the largest independent set in $M$.
\end{definition}

The simplest example is a $k$-uniform matroid in which a set $X \sub E$ is independent when $|X| \le k$. 
Another important example is a linear matroid.
Let $A$ be a matrix over a field $\mathbb{F}$. 
We define matroid $M = (E,\mathcal{I})$ where $E$ be the set of columns of $A$ and $X \sub E$ is independent in $M$ when the corresponding columns are independent over $\mathbb{F}$.
We say that the matrix $A$ is a representation of $M$ over $\mathbb{F}$.

Given a graph $G$, we define its graphic matroid $M = (E(G), \mathcal{I})$ where $X \sub E(G)$ is independent when $X$ does not contain a cycle.
It is well-known that every graphic matroid is linear and the oriented incidence matrix of $G$ forms a representation of $M$ over any field.

\begin{lemma}[\cite{oxley2006matroid}]
\label{lem:prelim:graphic}
Given a graph $G$ we can find a representation matrix of its graphic matroid over any field in polynomial time.
\end{lemma}

\begin{definition}[Product family]
Given two families of independent sets $\scal_1, \scal_2$ in a matroid $M = (E,\mathcal{I})$ we define
\[
\scal_1 \bullet \scal_2 = \{ X \cup Y \mid X \in \scal_1, Y \in \scal_2, X \cap Y = \emptyset, X \cup Y \in \mathcal I \}.
\]
\end{definition}

\subparagraph{Representative families.}

We say that a family of sets $\scal$ is a $p$-family if every set in $\scal$ has size~$p$.

\begin{definition}[Min/max $q$-representative family]
Let $M = (E,\mathcal{I})$ be a matroid, 
$\scal$ be a family of subsets of $E$, and $w \colon \scal \to \mathbb{N}$ be a non-negative weight function.
A subfamily $\widehat S \subseteq S$ is min
$q$-representative (resp. max $q$-representative) for $S$ if for every set $Y \subseteq E$ of size at most $q$, if there is a set $X \in S$ disjoint from $Y$ with $X \cup Y \in \mathcal{I}$, then there is a set $\widehat X \in \widehat S$ disjoint from $Y$ with (a) $\widehat X \cup Y \in \mathcal{I}$ and (b) $w(\widehat X) \le w(X)$ (resp. $w(\widehat X) \ge w(X)$).
\end{definition}

When all weights are zero, we obtain a simpler notion of a $q$-representative family.
Observe that when $X$ is a $p$-element set in a matroid of rank $k$ and $Y$ satisfies $X \cap Y = \emptyset$ and $X \cup Y \in \mathcal{I}$ then $|Y| \le k - p$.
We make note of this fact.

\begin{observation}\label{obs:prelim:repr:rank}
If $\scal$ is a $p$-family in a matroid of rank $k$ and $\widehat \scal \sub^{k-p}_{\mathrm{maxrep}} \scal$ then $\widehat \scal \sub^{k}_{\mathrm{maxrep}} \scal$.
\end{observation}

The following lemmas have been stated in \cite{representative-efficient} for the unweighted version of representative families but with a remark that they work as well for the weighted version (as stated below).

\begin{lemma}[{\cite[Lemma 3.1]{representative-efficient}}]
\label{lem:prelim:repr:transitive}
Let $M = (E,\mathcal{I})$ be a matroid and $\scal$ be a family of subsets of $E$.
If $\tilde\scal \sub^q_{\mathrm{maxrep}} \scal$ and $\widehat \scal \sub^q_{\mathrm{maxrep}} \tilde\scal$ then $\widehat \scal \sub^q_{\mathrm{maxrep}} \scal$.
\end{lemma}

\begin{lemma}[{\cite[Lemma 3.2]{representative-efficient}}]
\label{lem:prelim:repr:union}
Let $M = (E,\mathcal{I})$ be a matroid and $\scal$ be a family of subsets of $E$.
If $\scal = \scal_1 \cup \scal_2 \cup \dots \cup \scal_\ell$ and $\widehat \scal_i \sub^q_{\mathrm{maxrep}} \scal_i$, then $\bigcup_{i=1}^\ell \widehat\scal_i \sub^q _{\mathrm{maxrep}} \scal$.
\end{lemma}

\begin{lemma}[{\cite[Lemma 3.3]{representative-efficient}}]
\label{lem:prelim:repr:transitive-product}
Let $M = (E,\mathcal{I})$ be a matroid or rank $k$ and $\scal_1$ be a $p_1$-family of independent sets, $\scal_2$ be a $p_2$-family of independent sets, $\widehat \scal_1 \sub^{k-p_1}_{\mathrm{maxrep}} \scal_1$, $\widehat \scal_2 \sub^{k-p_2}_{\mathrm{maxrep}} \scal_2$.
Then $\widehat \scal_1 \bullet \widehat \scal_2 \sub^{k-p_1-p_2}_{\mathrm{maxrep}} \scal_1 \bullet \scal_2$.
\end{lemma}

The following theorem is the key to employ representative families in the design of single-exponential algorithms.
We state it only in the maximization variant.

\begin{theorem}[{\cite[Theorem 3]{representative-efficient}}]
\label{thm:prelim:repr:efficient}
Let $M = (E,\mathcal{I})$ be a linear matroid of rank $p+q=k$ given together with its representation matrix $A_M$ over a field $\mathbb{F}$.  Let $\scal$ be a $p$-family of independent sets in $M$.
Then a max $q$-representative family $\widehat S \subseteq S$ for $S$ with at most $k \choose p$ elements can be found in
$\Oh\left(|\scal|\cdot {k \choose p} \cdot p^\omega + |\scal|\cdot {k \choose p}^{\omega-1}\right)$
operations over $\mathbb{F}$.
\end{theorem}

We present a more concise corollary suited for our applications.

\begin{lemma}
\label{lem:prelim:repr:efficient-final}
Let $M = (E,\mathcal{I})$ be a graphic matroid of rank $k$.  Let $\mathcal{S}$ 
be a family of subsets of $E$. 
Then 
$\widehat \scal \sub^k_{\mathrm{maxrep}} \scal$ with at most $2^k$ elements can be found in time
$\Oh\left(|\scal| \cdot 2^{(\omega-1) k} \cdot k^\omega\right)$.
\end{lemma}
\begin{proof}
Thanks to \cref{lem:prelim:graphic} we can efficiently represent $M$ over $\mathbb{F}_2$.
For $p \in [k]$ let $\scal^p \sub \scal$ be the family of independent sets in $\scal$ of size $p$.
We apply \cref{thm:prelim:repr:efficient} to compute a max $(k-p)$-representative family $\widehat \scal^p \sub^{k-p}_{\mathrm{maxrep}} \scal^p$ of size at most $k \choose p$ for each $\scal_p$.
By \cref{obs:prelim:repr:rank} we can write $\widehat \scal^p \sub^{k}_{\mathrm{maxrep}} \scal^p$.
From \cref{lem:prelim:repr:union} we know that $\widehat \scal = \bigcup_{p = 1}^k \widehat \scal^p$ (plus $\emptyset$ if $\emptyset \in \scal$) is max $k$-representative for $\scal$.
The sizes of $\widehat \scal^p$ sum up to $2^k-1$ and the total running time can be upper bounded as in the statement with a trivial bound ${k \choose p} \le 2^k$.
\end{proof}

If the family $\scal$ has the special form of a product family, then instead of applying \cref{thm:prelim:repr:efficient} directly, one can obtain a slightly better running time.
Such families are of special importance for treewidth-based algorithm since they appear naturally in the computations for join nodes.
In the following theorem, the input consists of $\scal_1, \scal_2, \scal_1 \bullet \scal_2$, and the weight function $w \colon \scal_1 \bullet \scal_2 \to \mathbb{N}$, so it may have size $\Oh(4^k)$.

\begin{theorem}[{\cite[Corrolary 2]{representative-product}}]
\label{thm:prelim:repr:product}
Let $M = (E,\mathcal{I})$ be a linear matroid of rank $k$ given together with its representation matrix $A_M$ over a field $\mathbb{F}$.  
Let $\scal_1, \scal_2$ be two families of independent sets of $M$ and the number of sets of size $p$ in $\scal_1$ and $\scal_2$ be at most ${k+c} \choose p$.
Here, $c$ is a fixed constant.
Let $\scal_{r}^p$ be the subfamily of $\scal_r$ of sets of size $p$, for $r \in \{1,2\}$, $p \in [k]$.
Then for all pairs $p,q \in [k]$ we can find $\widehat \scal^{p,q} \sub^{k-p-q}_{\mathrm{maxrep}} \scal_1^p \bullet \scal_2^q$ of size $k \choose {p+q}$ in total
$\Oh\br{2^{(\omega - 1)k}3^k \cdot k^\omega}$
operations over $\mathbb{F}$.
\end{theorem}

We remark that in the original statement in \cite{representative-product} the number of operations is upper bounded by $\Oh\br{(2^\omega + 2)^k\cdot k^\omega + 2^{(\omega - 1)k}3^k \cdot k^\omega}$ but for every value of $\omega \ge 2$ it holds that $2^\omega + 2 \le 2^{\omega-1}\cdot 3$.

\begin{lemma}
\label{lem:prelim:repr:product-final}
Let $M = (E,\mathcal{I})$ be a graphic matroid of rank $k$.
Let $\scal_1, \scal_2$ be two families of independent sets of $M$, each of size at most $2^{k}$.
Then we can find $\widehat \scal \sub^k_{\mathrm{maxrep}} \scal_1 \bullet \scal_2$ of size at most $2^k$ in time
$\Oh\br{2^{(\omega - 1)k}3^k \cdot k^\omega}$.
\end{lemma}
\begin{proof}
As observed before, we can efficiently represent $M$ over $\mathbb{F}_2$ and we can assume that $\emptyset \not \in \scal_1 \bullet \scal_2$ as otherwise it can added in the end.
For $p \in [k], r \in \{1,2\}$, let $\scal^p_r \sub \scal_r$ be the family of independent sets in $\scal_r$ of size $p$.
We first apply \cref{thm:prelim:repr:efficient} for each pair $(r,p)$ to compute $\tilde \scal^p_r \sub^{k-p}_{\mathrm{maxrep}} \scal^p_r$ of size $k \choose p$.
The total running time is bounded by $\Oh(2^{\omega k} \cdot k^\omega)$ which is bounded by $\Oh\br{2^{(\omega - 1)k}3^k \cdot k^\omega}$.
Now we can apply \cref{thm:prelim:repr:product} to $\bigcup_{p=1}^k \tilde \scal^p_1$ and $\bigcup_{p=1}^k \tilde \scal^p_2$.
We obtain, for each pair $p,q \in [k]$, a family
$\widehat \scal^{p,q} \sub^{k-p-q}_{\mathrm{maxrep}} \tilde \scal_1^p \bullet \tilde \scal_2^q$.
By Lemmas~\ref{lem:prelim:repr:transitive}, \ref{lem:prelim:repr:transitive-product}, and \cref{obs:prelim:repr:rank}, we get that
$\widehat \scal^{p,q} \sub^{k}_{\mathrm{maxrep}} \scal_1^p \bullet  \scal_2^q$.
\cref{lem:prelim:repr:union} implies that $\widehat \scal' = \bigcup_{p,q \in [k]} \widehat \scal^{p,q} \sub^{k}_{\mathrm{maxrep}} \scal_1 \bullet \scal_2$.
The family $\widehat \scal'$ contains at most $k\cdot 2^k$ sets; we use \cref{lem:prelim:repr:efficient-final} to find $\widehat \scal \sub^{k}_{\mathrm{maxrep}}  \widehat \scal'$ of size at most $2^k$ in time $\Oh(2^{\omega k} \cdot k^{\omega+1})$ which is negligible compared to the running time from \cref{thm:prelim:repr:product}.
The claim follows from \cref{lem:prelim:repr:transitive}.
\end{proof}

}
\section{Chordal Deletion}
\label{sec:chordal}

We begin with a simple treewidth-preserving reduction from \textsc{Feedback Vertex Set}. 

\begin{restatable}[\starr]{lemma}{lemChLB}
\label{lem:chordal:LB}
Let $G$ be a graph and $\ell \in \mathbb{N}$.
Let $G'$ be obtained from $G$ by subdividing each edge.
Then $\tw(G') = \tw(G)$ and $G$ has a feedback vertex set (FVS) of size $\ell$ if and only if $G'$ has a chordal deletion set of size $\ell$.
\end{restatable}
\full{
\begin{proof}
When $S \sub V(G)$ is a FVS in $G$ then it is also a FVS in $G'$.
Since $G'-S$ is acyclic, it is also chordal.
In the second direction, consider a chordal deletion set $S' \sub V(G')$ in $G'$.
As $G'$ is bipartite, then $G'-S'$ is as well, so by \cref{lem:prelim:bip-chordal} it must be acyclic.
So $S'$ is a FVS in $G'$.
If $S'$ contains a vertex $w$ introduced by subdividing an edge $uv \in E(G)$ then every cycle in $G'$ going through $w$ also goes through $u$ and $v$.
Therefore $(S' \sm w) \cup \{u\}$ is also a FVS in $G'$.
We can thus assume that $S' \sub V(G)$ and it forms a FVS in $G$.

It remains to upper bound $\tw(G')$ as clearly $\tw(G') \ge \tw(G)$.
If $\tw(G) = 1$ then $G$ is a forest and so is $G'$.
Suppose that $\tw(G) \ge 2$ and consider a tree decomposition of $G$ of optimal width.
We can transform it into a tree decomposition of $G'$ as follows: for each $uv \in E(G)$ pick a node $t$ whose bag contains both $u, v$ and create a node $t_{uv}$, adjacent only to $t$, with a bag $\{u,v,w\}$, where $w$ is a vertex introduced on the edge $uv$.
Since all the created bags have size three, the width of the decomposition does not change.
\end{proof}
}

As a consequence, the base of the exponent $c$ in \cref{thm:intro:chordal} must be at least $3$ under Strong Exponential Time Hypothesis~\cite{cygan2011solving} and $c$ must be at least $2^\omega + 1$ if the current-best deterministic algorithm for \textsc{Feedback Vertex Set} parameterized by treewidth is optimal~\cite{wlodarczyk2019clifford}. While we have no evidence that the mentioned algorithm should be optimal for deterministic time, we provide this comparison to indicate that breaching this gap for \textsc{ChVD} would imply the same for a more heavily studied problem.

\subparagraph{Minimal vertex separators.}

We set the stage for the proof of \cref{prop:intro:homo}.
First we need to develop some theory about \minver{s} in chordal graphs.

\begin{definition}
Let $\texttt{MinSep}(G)$ denote the set of minimal vertex separators in a graph $G$.
For a graph $G$ and a (possibly empty) set $S \subseteq V(G)$, we define $\texttt{Comp}(G, S)$ to be the set of connected components $C_i$ of $G-S$ for which it holds that $N_G(C_i) = S$.
\end{definition}

Note that whenever $G$ is disconnected then $\emptyset \in \minsep(G)$
and
$\texttt{Comp}(G, \emptyset)$ is just the set of connected components of $G$.
According to \cref{lem:prelim:separator-minimal-comp}, the set $S$ is a minimal $(u,v)$-separator if and only if $u,v$ belong to some (distinct) components from $\comp(G,S)$.
For later use, we establish a relation between sets $\minsep(G)$, $\comp(G,S)$ in $G$ and a graph obtained by a removal of a simplicial vertex.

\begin{lemma}[\starr]\label{lem:chordal:separator-recurse}
Let $v$ be a simplicial vertex in $G$
and $S \in \minsep(G)$.
If $S \ne N_G(v)$ then $S \in \minsep(G-v)$ and $|\comp(G,S)| = |\comp(G-v,S)|$.
\end{lemma}
\full{
\begin{proof}
Suppose that $S \ne N_G(v)$.
By \cref{lem:prelim:separator-minimal-simplicial} we know that $v \not\in S$.

First, consider the case $N_G(v) \subsetneq S$.
Then $\{v\}$ forms a connected component of $G-S$ but
$\{v\} \not\in \comp(G,S)$.
Next, \cref{lem:prelim:separator-minimal-comp} implies that $S$ is not a minimal $(v,u)$-separator for any $u \in V(G)$.
Therefore $S \in \minsep(G-v)$
and $|\comp(G,S)| = |\comp(G-v,S)|$.

In the second case $N_G(v) \not\subseteq S$.
Let $u \in N_G(v) \setminus S$ and $C$ be the connected component of $G-S$ which contains $v$.
Then $u \in V(C)$.
Since $v$ is simplicial, we have $N_G(v) \subseteq N_G[u]$.
Therefore, $C - v$ is connected, $N_{G-v}(C \sm v) = N_G(C)$, and so inserting $v$ to $(G-v)-S$ does not affect the number of connected components nor their neighborhoods.
This means that $S \in \minsep(G-v)$ and $|\comp(G,S)| = |\comp(G-v,S)|$.
\end{proof}}

We need a simple technical lemma about minimal vertex separators.

\begin{lemma}[\starr]
\label{lem:separator-many-sets}
Let $G$ be a connected graph and $V_1, \dots, V_k \subseteq V(G)$, $k \ge 2$, be disjoint sets so that $G[V_i]$ is connected, for $i \in [k]$, and $E_G(V_i, V_j) = \emptyset$, for $i\ne j$.
Then there exists a minimal vertex separator $S \sub V(G) \sm (V_1 \cup \dots \cup V_k)$ in $G$
which is a $(V_i,V_j)$-separator for some $i \ne j$ and
each set $V_i$ is contained in some component $C \in \comp(G,S)$.
\end{lemma} 
\full{
\begin{proof} 
Let $S \sub V(G) \sm (V_1 \cup \dots \cup V_k)$ be an inclusion-minimal set
with the following property: $S$~separates sets $V_i, V_j$ for some $i \ne j$.
Such a set $S$ must exist because $N_G(V_1)$ has this property.

We argue that $S$ satisfies the conditions of the lemma.
Clearly $S$ is a minimal $(v_i, v_j)$-separator for each $v_i \in V_i$, $v_j \in V_j$.
Suppose that for some $h \in [k]$ the set $V_h$ is not contained in any component from $\comp(G,S)$.
Let $C$ be the component of $G-S$ that contains $V_h$.
Since $C \not\in  \comp(G,S)$, we have $N_G(C) \subsetneq S$.
At least one of the sets $V_i, V_j$ is not contained in $C$; assume w.l.o.g. that it is $V_i$.
Then $N_G(C)$ is a $(V_h,V_i)$-separator being a proper subset of $S$, which contradicts the choice of $S$.
The claim follows.
\end{proof}
}


\short{
We will use the following concept which appears in the algorithm for \textsc{ChVD} by Jansen~et~al. In the full version, we also provide several properties of this operation, used to process partial solutions in a treewidth DP.}
\full{We will use the following concept which appears in the algorithm for \textsc{ChVD} by Jansen~et~al.}

\begin{definition}[{\cite[Def. 5.55]{JansenKW21arxiv}}]
\label{def:prelim:condense}
For a graph $G$ and a vertex set $X \subseteq V(G)$ let the graph $\compress(G,X)$ be obtained from $G$ by contracting the connected connected components of $G-X$ into single vertices and then removing those of them which are simplicial.
\end{definition}

\short{
\refstepcounter{theorem}
\refstepcounter{theorem}
\refstepcounter{theorem}
}
\full{
We say that $G$ is {\em condensed} with respect to $X$ if $G = \compress(\widehat G,X) $ for some graph $\widehat G$ or,
equivalently,  $G = \compress( G,X)$.
Due to the following facts, condensation forms a handy tool for efficiently storing partial solutions for \textsc{ChVD}.

\begin{lemma}[{\cite[Lem. 5.57]{JansenKW21arxiv}}]
\label{lem:chordal:condense-glue}
Consider compatible boundaried graphs $(G,X)$, $(H,X)$ so that $G,H$ are chordal.
Let $\widehat G = \compress(G,X)$.
Then $(G,X) \oplus (H,X)$ is chordal if and only if  $(\widehat G,X) \oplus (H,X)$ is chordal.
\end{lemma}

\begin{lemma}[{\cite[Lem. 5.53]{JansenKW21arxiv}}]
\label{lem:chordal:condense-size}
Consider a chordal graph $G$ with a non-empty vertex subset $X \sub V(G)$.
If $G$ is condensed with respect to $X$, then $|V(G)| \le 2|X| - 1$.
\end{lemma}

\begin{observation}
\label{lem:chordal:condense-assocciate}
Consider compatible boundaried graphs $(G,X)$, $(H,X)$.
Let $\widehat G = \compress(G,X)$ and  $\widehat H = \compress(H,X)$.
Then $\compress((G,X) \oplus (H,X),X) = (\widehat G,X) \oplus (\widehat H,X)$.
\end{observation} }

In this section we will exploit the following property of condensation.

\begin{restatable}[\starr]{lemma}{lemChCritOld}
\label{lem:chordal:criterion-old}
Consider a graph $G$ with a vertex set $X$ so that $G[X]$ is chordal. 
Then $G$ is chordal if and only if the following conditions hold:
\begin{enumerate}[nolistsep]
    \item for each connected component $C$ of $G-X$ the graph $G[X \cup C]$ is chordal,
    \item the graph $\compress(G,X)$ is chordal. 
\end{enumerate}
\end{restatable}
\full{
\begin{proof}
The forward direction is clear as the class of chordal graph is closed under vertex deletions and edge contractions.
We prove the opposite direction by induction on the number $k$ of the connected components in $G-X$.
For $k=1$ the condition (1) suffices to obtain chordality of $G$.
Suppose now that $k > 1$ and consider a partition $V(G) \sm X = A \cup B$ where $A$ induces a single connected component of $G-X$ and $B$ induces the rest of them.
Let $\widehat G_A = \compress(G[A \cup X], X)$ and $\widehat G_B = \compress(G[B \cup X], X)$.
From \cref{lem:chordal:condense-assocciate} we know that $(\widehat G_A, X) \oplus (\widehat G_B, X) = \compress(G,X)$;
in particular this implies that $\widehat G_A, \widehat G_B$ are chordal.
From inductive assumption we get that $G[A \cup X], G[B \cup X]$ are chordal.
We apply \cref{lem:chordal:condense-glue} (twice) to obtain that $G = (G[A \cup X],X) \oplus (G[B \cup X],X)$ is chordal as well.
\end{proof}
}

In order to turn \cref{lem:chordal:criterion-old} into a more convenient criterion, we will compress information about a graph $G$ with a vertex subset $X$ into multiple auxiliary graphs, one for each \minver in $G[X]$.

\begin{definition}\label{def:chordal:aux}
Consider a graph $G$ with a vertex set $X$ so that $G[X]$ is chordal.
For a set $S \in \texttt{MinSep}(G[X])$ we construct the graph $\texttt{Aux}(G,X,S)$ as follows:

\begin{enumerate}[nolistsep]
    \item contract each $C \in \texttt{Comp}(G[X], S)$ into a vertex and remove the remaining vertices of~$X$ (including all of $S$),
    \item contract each connected component of $G-X$ into a vertex.
\end{enumerate}
\end{definition}

Note that $\texttt{Aux}(G,X,\emptyset)$ is obtained by just contracting each connected component of $G[X]$ and each connected component of $G-X$.
Moreover, observe that $\texttt{Aux}(G,X,S)$ is always a bipartite graph because there can be no edges between two components from   $\texttt{Comp}(G[X], S)$ nor between two components of $G-X$.
See \Cref{fig:aux} for an example of this construction.

\begin{figure}
    \centering
    \includegraphics[width=0.8\linewidth]{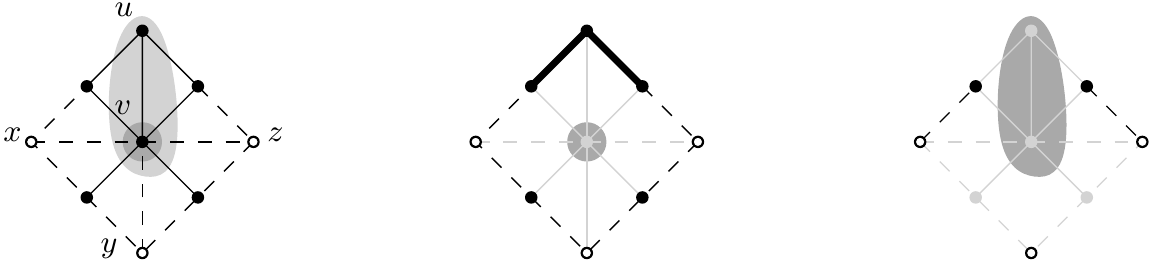}
    \caption{On the left: graph $G$ and set $X \sub V(G)$ represented by black disks. The graph $G[X]$ is drawn with solid edges.
    There are two \minver{s} in $G[X]$: $S_1 = \{v\}$ and $S_2 = \{u,v\}$, sketched in gray.
    In the middle: the graph $\aux(G,X,S_1)$ with thick edges indicating a component that gets contracted into a single vertex;
    the gray vertices and edges are removed.
    On the right: the graph $\aux(G,X,S_2)$; note that $|\comp(G[X],S_2)|=2$ because the lower vertices of $X$ are not adjacent to every vertex in $S_2$.
    The graph  $\aux(G,X,S_1)$ contains a cycle and this witnesses that $G$ is not chordal.
    However, removing from $G$ any single vertex among $x,y,z$ results in a chordal graph.
    }
    \label{fig:aux}
\end{figure}

To make a connection between holes in $G$ and cycles in $\aux(G,X,S)$, we need a criterion to derive existence of a cycle from a closed walk with certain properties.
In the following lemma we consider a cyclic order on a sequence of length $k$.
We define the successor operator as $s(i) = i+1$, for $i \in [k-1]$, and $s(k) = 1$. 

\begin{restatable}[\starr]{lemma}{lemChBipWalk}
\label{lem:bipartite-walk}
Let $G$ be a bipartite graph with vertex partition $V(G) = A \cup B$.
Suppose there exists a sequence of vertices $(v_1, \dots, v_k)$ in $G$ such that:
\begin{enumerate}[nolistsep]
    \item for $i \in [k]$ it holds $v_i = v_{s(i)}$ or $v_iv_{s(i)} \in E(G)$,
    \item the multiset $\{v_1, \dots, v_{k}\}$ contains at most one occurrence of each vertex from $A$,
    \item the set $\{v_1, \dots, v_k\}$ contains at least two vertices from $B$.
\end{enumerate} 
Then $G$ contains a cycle.
\end{restatable}
\full{
\begin{proof}
We apply modifications to the sequence $(v_1, \dots, v_k)$ while preserving conditions (1-3).
First, if $v_i = v_{s(i)}$ then remove $v_{s(i)}$.
This rule is clearly safe.
Second, if $v_i = v_{s(s(i))} \ne v_{s(i)}$ then remove $v_{s(i)}$ and $v_{s(s(i))}$.
Due to condition (2) it must be 
$v_i \in B$ and $v_{s(i)} \in A$ so the set  $\{v_1, \dots, v_k\} \cap B$ stays invariant, which preserves condition (3).

Each modification shortens the sequence, so after applying them exhaustively we obtain a sequence that cannot be further reduced.
Due to condition (3) the length of the sequence cannot drop below 4.
We claim that each edge from $E(G)$ is now traversed at most once. 
Suppose otherwise that $v_iv_{s(i)}, v_jv_{s(j)}$ represent the same edge for $i \ne j$.
The indices $i,j$ cannot be consecutive
due to the second modification rule.
But then some vertex of $A$ must occur twice in the sequence which contradicts condition (2).
As a result we obtain a non-trivial closed walk in $G$ without repeated edges, which implies the existence of a~cycle.
\end{proof}
}

{We are ready to prove a proposition creating a link between chordality and acyclicity.}

\begin{proposition}\label{lem:chordal:criterion1}
Consider a graph $G$ with a vertex subset $X \sub V(G)$ so that for each connected component $C$ of $G-X$ the graph $G[X \cup C]$ is chordal.
Then $G$ is chordal if and only if for each $S \in \texttt{MinSep}(G[X])$ the graph $\texttt{Aux}(G,X,S)$ is acyclic.
\end{proposition}
\begin{proof}
First we argue that if $G$ is chordal then all graphs $\texttt{Aux}(G,X,S)$ are acyclic.
Because the class of chordal graphs is closed under vertex deletions and edge contractions, the graphs $\aux(G,X,S)$ are chordal as well.
Since each graph $\aux(G,X,S)$ is also bipartite,
by \cref{lem:prelim:bip-chordal} we obtain that $\aux(G,X,S)$ is acyclic.

Now suppose that $G$ is not chordal.
Let $G' = \compress(G,X)$ (recall \cref{def:prelim:condense}).
By~\cref{lem:chordal:criterion-old}, the graph $G'$ is not chordal as well but for each vertex $v \in V(G') \setminus X$ the graph $G'[X \cup \{v\}]$ is chordal (because contraction preserves chordality).
Note that 
$\texttt{Aux}(G',X,S)$ is an induced subgraph of $\aux(G,X,S)$ for each $S \in \texttt{MinSep}(G[X])$ (they may differ only due to removal of simplicial vertices), so it suffices to show that
one of the graphs $\texttt{Aux}(G',X,S)$ has a cycle.

As $G'$ is not chordal, it contains a hole $H = (u_1, \dots, u_k)$. 
We consider two cases: either $V(H)$ intersects at least two connected components of $G'[X]$ or only one.
In the first case, let $\phi_0 \colon V(G') \to V(\texttt{Aux}(G',X,\emptyset))$ be the mapping given by the contractions from \cref{def:chordal:aux}.
Recall that $V(G') \sm X$ is an independent set in $G'$ so $\phi_0$ is an identity on this set.
The sequence $(\phi_0(u_1), \dots, \phi_0(u_k))$ meets the preconditions of \cref{lem:bipartite-walk} for $A = V(G') \sm X$ and $B = \phi_0(X)$ so $\texttt{Aux}(G',X,\emptyset)$ has a cycle.
As $G'[X] = G[X]$ is disconnected, we have $\emptyset \in \minsep(G[X])$.

In the second case, let $Y \subseteq X$ induce the only connected component of $G'[X]$ that intersects $V(H)$.
Let $V_1, \dots, V_\ell \subseteq Y$ be the vertex sets of maximal subpaths of $H$ within $Y$.
By the definition of a hole, we have $E_{G'}(V_i,V_j)=\emptyset$ for distinct $i, j \in [\ell]$.
It must be $\ell \ge 2$ because for each $v \in V(G') \sm X$ the graph $G'[X \cup \{v\}]$ is chordal and the hole $H$ must visit at least two vertices from the independent set $V(G') \setminus X$.
By \cref{lem:separator-many-sets}, there exists a minimal vertex separator $S \subseteq Y \setminus V(H)$ in $G'[Y]$ such that 
every set $V_i$ is contained in some component from $\comp(G'[Y], S)$ and at least two components from $\comp(G'[Y], S)$ intersect $V(H)$.
Note that $S \in \minsep(G[X])$.
Let $C_S$ be the union of the components from  $\comp(G'[Y], S)$; note that $V(H) \subseteq V(C_S) \cup (V(G') \sm X)$.

Let $\phi_S \colon V(C_S) \cup (V(G') \sm X) \to V(\texttt{Aux}(G',X,S))$ be the mapping given by the contractions from \cref{def:chordal:aux} which turn each component from $\texttt{Comp}(G'[Y], S)$ into a single vertex.
Again, the sequence $(\phi_S(u_1), \dots, \phi_S(u_k))$ meets the preconditions of \cref{lem:bipartite-walk} for $A = V(G') \sm X$ and $B = \phi_S(V(C_S))$ so $\texttt{Aux}(G',X,S)$ has a cycle.
See \Cref{fig:aux} for an illustration.
\end{proof}

\full{
Observe that whenever a component of $G-X$ is simplicial then in every graph $\aux(G,X,S)$ the corresponding vertex has degree one and so it cannot be a part of any cycle.
Therefore the simplicial components of $G-X$ does not affect the criterion from \cref{lem:chordal:criterion1}.
This agrees with the definition of $\compress(G,X)$ where the simplicial components are removed as meaningless.}

\subparagraph{Signatures of boundaried graphs.}

The next step is to construct a graphic matroid $M_B$ for a chordal graph $B$ so that for any two graphs $G_1, G_2 \in \gcal_{X,B}$ 
the information about chordality of $(G_1,X) \oplus (G_2,X)$ could be read from~$M_B$.
\cref{lem:chordal:criterion1} already relates chordality to acyclicity but the corresponding graphic matroids for $G_1, G_2$ are disparate.
To circumvent this, we will further compress the information about cycles.

\begin{definition}
Consider a graph $B$.
For $S \in \minsep(B)$, let $\base(B,S)$
be the complete graph on vertex set $\comp(B,S)$. 
The graph $\base(B)$ is a disjoint union of  all the graphs $\base(B,S)$ for $S \in \minsep(B)$. 
\end{definition}

That is, we treat the components from $\comp(B,S)$ as abstract vertices of a new graph 
which is a union of cliques.

The following transformation is similar to the one used in the algorithm for \textsc{Steiner Tree}
based on representative families~\cite{representative-efficient}.
For the sake of disambiguation, in the  definition below we assume an implicit linear order on the vertices of $B$; this order may be arbitrary.
Since vertices of $\base(B)$ correspond to distinct subsets of $V(B)$, which can ordered lexicographically, fixing the order on $V(B)$ yields an order on $V(\base(B))$.
We can thus assume that also the vertices of $V(\base(B))$ are linearly ordered. 

\begin{definition}\label{def:chordal:span}
Consider a chordal graph $B$ and 
$Y \subseteq V(B)$.
We define the \emph{spanning signature} $\spann(B,Y) \subseteq E(\base(B))$ as follows. 
For each $S \in \minsep(B)$ let $C_{S,Y} \subseteq V(\base(B,S))$ be given by components from $\comp(B,S)$ with a non-empty intersection with~$Y$.
Let $P_{S,Y} \subseteq E(\base(B,S))$ be the path connecting the vertices of $C_{S,Y}$ in the increasing order.
Then $\spann(B,Y) = \bigcup_{S \in \minsep(B)} P_{S,Y}$.
\end{definition}

In other words, $\spann(B,Y)$ is a disjoint union of paths in the graph $\base(B)$, where each path encodes the relation between $Y$ and a respective \minver{}~in~$B$.

The next lemma states that under certain conditions replacing a vertex $v$ with a 
tree over $N(v)$ (in particular: a~path) does not affect acyclicity of the graph.
Note that due to the precondition $|N(u) \cap N(v)| \le 1$ 
we never attempt to insert an edge that is already present.

\begin{lemma}[\starr]
\label{lem:chordal:bip-transform}
Let $G$ be a bipartite graph with a vertex partition $V(G) = A \cup B$ so that for each distinct $u,v \in A$ it holds that $|N_G(u) \cap N_G(v)| \le 1$.
Consider a graph $G'$ obtained from $G$ by replacing each vertex $v \in A$ by an arbitrary tree on vertex set $N_G(v)$.
Then $G$ is acyclic if and only if $G'$ is acyclic.
\end{lemma}
\full{
\begin{proof}
For a graph $G$, let $\mu(G)$ be the multiset of integers $(|V(C)| - |E(C)|)_{C \in \mathcal{C}}$
where $\mathcal{C}$ is the family of connected components of $G$.
A graph $G$ is acyclic if and only if $\mu(G)$ contains only 1's.
We show that the described modifications, performed in an arbitrary order, does not affect $\mu(G)$ except for possibly removing some 1's.
Let $v \in V(G)$, $d = |N_G(v)|$, and $C$ denote the connected component of $v$.
If $v$ is isolated, then removing $v$ translates into removing single 1 from $\mu(G)$.
Otherwise, replacing $v$ with a tree on $N_G(v)$ transforms $C$ into another connected graph $C'$.
We remove one vertex and $d$ edges, so $|V(C)| - |E(C)|$ drops by $d-1$.
On the other hand, any tree over $N_G(C)$ has exactly $d-1$ edges.
Due to the assumption $|N_G(u) \cap N_G(v)| \le 1$ for $u \ne v$, every inserted tree is disjoint from previously inserted edges among $B$.
Hence, we insert exactly $d-1$ new edges and $|V(C)| - |E(C)| = |V(C')| - |E(C')|$.
We also only remove vertices from $A$ so we never remove an endpoint of an inserted edge.
The claim follows by observing that $\mu(G)$ contains an element different from 1 if and only if $\mu(G')$ does.
\end{proof}
}

This allows us to translate the criterion from \cref{lem:chordal:criterion1} into a more convenient one, in which the vertex set of the auxiliary graph depends only on $G[X]$ rather than $G$.

\begin{lemma}\label{lem:chordal:sign-proof}
Consider a graph $G$ with a vertex subset $X \sub V(G)$.  
Let $\mathcal{C}$ denote the family of connected components of $G-X$.
Suppose that for each $C \in \mathcal{C}$ the graph $G[X \cup C]$ is chordal.
Then $G$ is chordal if and only if:
\begin{enumerate}[nolistsep]
    \item the sets $\spann(G[X], N_G(C))$, for different $C \in \mathcal{C}$, are pairwise disjoint, 
    \item the union of sets $\spann(G[X], N_G(C))$, over $C \in \mathcal{C}$, forms an acyclic edge set in $E(\texttt{Base}(G[X]))$.
\end{enumerate}
\end{lemma}
\begin{proof}
From \cref{lem:chordal:criterion1}
we know that $G$ is chordal if and only if for each $S \in \texttt{MinSep}(G[X]) $ the graph $\texttt{Aux}(G,X,S)$ is acyclic.
We consider two cases.

First, suppose that for some $S \in \texttt{MinSep}(G[X]) $ there are two vertices representing distinct components $C_1,C_2 \in \mathcal{C}$ that share two common neighbors $x,y$ in $\texttt{Aux}(G,X,S)$. 
In other words, there are two components from $\comp(G[X],S)$ that intersect both $N_G(C_1)$ and $N_G(C_2)$.
Then $\texttt{Aux}(G,X,S)$ contains a cycle of length 4, so $G$ is not chordal. 
If  $\spann(G[X], N_G(C_1))$ and  $\spann(G[X], N_G(C_2))$ share an edge, then condition (1) fails, so suppose this is not the case.
But then the paths $P_{S,N(C_1)}$ and $P_{S,N(C_2)}$ (recall \cref{def:chordal:span}) are edge-disjoint and they both visit $x$ and $y$.
As a consequence, $x,y$ lie on a cycle contained in the edge set $\spann(G[X], N_G(C_1)) \cup  \spann(G[X], N_G(C_2))$ so condition (2) fails.
In summary, both $G$ is not chordal and one of conditions (1, 2) does~not~hold.

Next, suppose that for each $S \in \texttt{MinSep}(G[X]) $ and any two vertices representing distinct components $C_1,C_2 \in \mathcal{C}$
the intersection of their neighborhoods in $\texttt{Aux}(G,X,S)$ contains at most one element.
This implies condition (1).
Consider a graph $H$ given by a~disjoint union of all graphs $\texttt{Aux}(G,X,S)$ over $S \in \texttt{MinSep}(G[X]) $.
This graph meets the preconditions of \cref{lem:chordal:bip-transform}.
Replacing each $\mathcal{C}\text{-component}$-vertex in $\texttt{Aux}(G,X,S)$ by the path $P_{S,N(C)}$ 
transforms $H$ into a subgraph of $\base(G[X])$ with the edge set $\bigcup_{C \in \mathcal{C}} \spann(G[X], N_G(C))$.
By~\cref{lem:chordal:bip-transform}, this graph is acyclic if and only if the graph $H$ is.
By~\cref{lem:chordal:criterion1}, this condition is equivalent to $G$ being chordal.
The lemma~follows.
\end{proof}

{We are ready to define the graphic matroid encoding all the necessary information about where a~hole can appear after gluing two chordal graphs.}
Recall that a graphic matroid of a graph $G$ is a set system over $E(G)$ where a subset $S \sub E(G)$ is called {independent} when $S$ contains no cycles.
\short{More information about matroids can be found in the preliminaries of the full version of the article.}

\begin{definition}
For a graph $B$ on vertex set $X$ we define matroid $M_B$ as the graphic matroid of the graph $\base(B)$.
For a graph $G \in \mathcal{G}_{X,B}$
the signature  $\texttt{Sign}(G,X) \subseteq E(\base(B))$ is defined as a union of $\spann(B, N_G(C))$ over all connected components $C$ of $G-X$.
\end{definition}

It follows from \cref{lem:chordal:sign-proof} that whenever $G$ is chordal then $\sign(G,X)$ is acyclic and so it forms an independent set in the matroid $M_{G[X]}$.
We can now give the existential part of \cref{prop:intro:homo}.
The mapping $\sigma \colon \mathcal G_{X,B} \to 2^{E(M_B)}$ therein is given here as $\sigma(G) = \sign(G,X)$. 

\begin{lemma}[\starr]\label{lem:chordal:criterion}
Let $(G_1,X)$ and $(G_2,X)$ be compatible boundaried chordal graphs.
Then $G = (G_1,X) \oplus (G_2,X)$ is chordal if and only if the sets $\texttt{Sign}(G_1,X)$, $\texttt{Sign}(G_2,X) \sub E(\base(G[X]))$ are disjoint and $\texttt{Sign}(G_1,X) \cup \texttt{Sign}(G_2,X)$ is acyclic.

Furthermore, $\sign(G,X) = \texttt{Sign}(G_1,X) \cup \texttt{Sign}(G_2,X)$.
\end{lemma}
\full{
\begin{proof}
Let $C_1, C_2, \dots, C_\ell$ denote the connected components of $G_1 - X$ and 
 $D_1, D_2, \dots, D_r$ denote the connected components of $G_2 - X$.
Clearly all graphs $G_1[X \cup C_i]$ and $G_2[X \cup D_i]$ are chordal.
Let $\mathcal{S}_1$ be the family of sets $\{\spann(G[X], N_{G_1}(C_i)\}_{i=1}^\ell$
and $\mathcal{S}_2$ be $\{\spann(G[X], N_{G_2}(D_i)\}_{i=1}^r$.
It follows from \cref{lem:chordal:sign-proof} that the sets in $\mathcal{S}_1$
are pairwise disjoint and their union, which is $\sign(G_1,X)$, is an acyclic edge set in $E(\base(G[X]))$.
The same holds for $\mathcal{S}_2$ and $\sign(G_2,X)$.
Again by \cref{lem:chordal:sign-proof}, the graph $G = (G_1,X) \oplus (G_2,X)$ is chordal if and only if the sets in the family $\mathcal{S}_1 \cup \mathcal{S}_2$ are pairwise disjoint and their sum is acyclic.
This is equivalent to the condition that $\sign(G_1,X), \sign(G_2,X)$ are disjoint and $\sign(G_1,X) \cup \sign(G_2,X)$ is acyclic, as intended.
By definition, $\sign(G,X)$ is the union of $\spann(G[X], N_{G}(C))$ over all connected components $C$ of $G - X$.
This union equals $\sign(G_1,X) \cup \sign(G_2,X)$.
\end{proof} }

The following lemma is the main ingredient in the running time analysis.
As the bound on the representative family's size is exponential in the rank of a matroid\footnote{We remark that Fomin et al.~\cite{representative-efficient} also considered a case when the rank might be large and the exponential term is governed by a different parameter but it is not applicable in our case.},
it is necessary to bound the rank~of~$M_B$.
It is known that the number of minimal vertex separators in a chordal graph is bounded by the number of vertices but we need a strengthening of  this fact. 


\begin{lemma}\label{lem:chordal:rank}
For a non-empty chordal graph $B$, the rank of $M_B$ is at most $|V(B)| - 1$.
\end{lemma}
\begin{proof}
Let $k = |V(B)|$.
The rank of $M_B$ equals the size of a spanning forest in $\base(B)$.
The vertex sets of connected components of $\base(B)$ are the sets $\comp(B,S)$ for $S \in \minsep(B)$.
Therefore it suffices to estimate
\[
\sum_{S \in \minsep(B)} (|\comp(B,S)| - 1) \le k - 1.
\]

We first prove the inequality for connected chordal graphs by induction on $k$.
For $k = 1$ the sum is zero.
Consider $k > 1$.
By \cref{lem:prelim:simplicial-exists}, $B$ contains a simplicial vertex.
Let $v$ be a simplicial vertex in $B$ and suppose that the claim holds for the graph $B-v$ (which is connected).
Let $S$ be a \minver in $B$.
By \cref{lem:chordal:separator-recurse} when 
$S \ne N_B(v)$ then $S \in \minsep(B-v)$
and $|\comp(B,S)| = |\comp(B-v,S)|$.
In that case the summand coming from $S$ is the same for $B$ and $B-v$.

It remains to handle the case $S = N_B(v)$.
Clearly, $\{v\} \in \comp(B,S)$.
If $|\comp(B,S)| = 1$ then $S \not\in \minsep(B)$ (\cref{lem:prelim:separator-minimal-comp}).
If $|\comp(B,S)| = 2$ then $S \in \minsep(B) \sm \minsep(B-v)$ and the sum grows by one.
If $|\comp(B,S)| \ge 3$ then $S \in \minsep(B) \cap \minsep(B-v)$ and $|\comp(B,S)| = |\comp(B-v,S)|+1$
so the sum again grows by one.
This concludes the proof of the inequality for connected chordal graphs.

When $B$ is disconnected, let $B_1, B_2, \dots, B_t$ denote its connected components and let $k_i = |V(B_i)|$. 
We have $|\comp(B,\emptyset)|-1 = t-1$.
Together with the sums for $B_1, B_2, \dots, B_t$ the total sum is at most $\sum_{i=1}^t k_i -t + t-1 = k - 1$.
\end{proof}

The last thing to be checked is whether we can compute the signatures efficiently.
To this end, we enumerate minimal vertex separators using 
\cref{lem:chordal:separator-recurse}.

\begin{lemma}[\starr]\label{lem:chordal:sign-poly}
There is a polynomial-time algorithm that, given a graph $G$ with a vertex subset $X \sub V(G)$ such that $G[X]$ is chordal, computes $\sign(G,X)$.
\end{lemma}
\full{
\begin{proof}
Let $B = G[X]$.
We show that one can enumerate $\minsep(B)$ in polynomial time.
By \cref{lem:prelim:simplicial-exists} the graph $B$ contains a simplicial vertex $v$.
This vertex can be found in polynomial time.
We recursively enumerate $\minsep(B-v)$.
By \cref{lem:chordal:separator-recurse}, if $S \in \minsep(B)$ then either $S \in \minsep(B-v)$ or $S = N_B(v)$, so the output size increases by at most one.
We can verify which elements of $\minsep(B-v) \cup \{N_B(v)\}$ are minimal vertex separators in $G$ using \cref{lem:prelim:separator-minimal-comp};
as a byproduct we obtain the sets $\comp(B, S)$.
It~remains to directly follow the definition of $\sign(G,X)$.
\end{proof} }

Lemmas \ref{lem:chordal:criterion}, \ref{lem:chordal:rank}, and \ref{lem:chordal:sign-poly} entail \cref{prop:intro:homo} but instead of working with that abstract statement we use these three lemmas directly when describing the final algorithm. 
\short{The results of this section allow us to employ the framework of representative families in order to truncate the number of partial solutions stored at a node of a tree decomposition to $2^{\Oh(\tw)}$.
The dynamic programming algorithm follows the lines of proofs in~\cite{representative-product} and is described in detail in the full version.
The main technical hurdle comes from the necessity to store only the condensed counterparts of the partial solutions.
The condensed graphs have only $\Oh(\tw)$ vertices each, what is the key to obtain a linear dependency on $|V(G)|$.
}

\full{
\subparagraph{Representative families for boundaried graphs.}

We  translate the framework of representative families from the language of matroids to chordal graphs and gluing.

\begin{definition}
Consider a family of chordal graphs
$\mathcal{G} \sub \mathcal{G}_{X,B}$ for some pair $(X,B)$
and a non-negative weight function $w: \mathcal{G} \to \mathbb{N}$.
We say that a subfamily $\widehat {\mathcal{G}} \sub \mathcal{G}$ is max-representative for $\mathcal{G}$ (and write $\widehat \gcal \sub_{\mathrm{maxrep}} \gcal$) if the following holds.
For every graph $H \in \mathcal{G}_{X,B}$, if there exist $G \in \mathcal{G}$ so that $(H,X) \oplus (G,X)$ is chordal, then there exists
$\widehat G \in \widehat{\mathcal{G}}$ so that  $(H,X) \oplus (\widehat G,X)$ is chordal and $w(\widehat G) \ge w(G)$.
\end{definition}

\begin{lemma}\label{lem:chordal:repr-from-graphic}
Consider a family of chordal graphs
$\mathcal{G} \sub \mathcal{G}_{X,B}$ for some pair $(X,B)$
and a non-negative weight function $w: \mathcal{G} \to \mathbb{N}$. 
Suppose that the matroid $M_B$ has rank $r$.
Let $\scal = \{\sign(G,X) \mid G \in \mathcal{G}\}$ and $\tau \colon \scal \to \mathcal{G}$ be given as $\tau(Y) = \mathrm{argmax}\, \{w(G) \mid G \in \mathcal{G},\, \sign(G,X) = Y\}$.
Suppose that $\widehat \scal \sub^r_{\mathrm{maxrep}} \scal$ with respect to matroid $M_B$ and weight function $w_\scal(Y) = w(\tau(Y))$.
Then $\tau(\widehat \scal) \sub_{\mathrm{maxrep}} \mathcal{G}$.
\end{lemma}
\begin{proof}
Let $H \in \mathcal{G}_{X,B}$ and $G \in \mathcal{G}$ be such that $(H,X) \oplus (G,X)$ is chordal.
Clearly $H$ must be chordal as well.
The set $S = \texttt{Sign}(G,X)$ belongs to $\scal$ and $w_\scal(S) \ge w(G)$.
By \cref{lem:chordal:criterion} we have that
$\texttt{Sign}(H,X) \cap S= \emptyset$ and $\texttt{Sign}(H,X) \cup S$ is acyclic.
By the definition of a max-representative family for a graphic matroid, there exists a set 
$\widehat S \in \widehat \scal$ so that 
$\texttt{Sign}(H,X) \cap \widehat S= \emptyset$, $\texttt{Sign}(H,X) \cup \widehat S$ is acyclic, and $w_\scal(\widehat S) \ge w_\scal(S)$.
Let $\widehat G = \tau(\widehat S)$.
Again by  \cref{lem:chordal:criterion} we infer that $(H,X) \oplus (\widehat G,X)$ is chordal.
Finally, $w(\widehat G) = w_\scal(\widehat S) \ge w_\scal(S) \ge w(G)$.
\end{proof}

\begin{lemma}
\label{lem:tw:repr:efficient}
Consider a family of chordal graphs
$\mathcal{G} \sub \mathcal{G}_{X,B}$ for some pair $(X,B)$
and a non-negative weight function $w: \mathcal{G} \to \mathbb{N}$.
Suppose that $|X| = k > 0$, $|\mathcal G| \le 2^{k+c}$, and each $G \in \mathcal{G}$ has at most $ck$ vertices.
Here, $c$ is a fixed constant.
Then a max-representative family $\widehat G \sub \widehat{\mathcal{G}}$ for $\mathcal{G}$ of size at most $2^{k-1}$ can be computed in time $\Oh(2^{\omega k}\cdot k^{\omega})$.
\end{lemma}
\begin{proof}
Let $\scal = \{\sign(G,X) \mid G \in \mathcal{G}\}$.
This family can be computed in time $2^k\cdot k^{\Oh(1)}$ thanks to \cref{lem:chordal:sign-poly}.
By \cref{lem:chordal:rank} the rank $r$ of $M_B$ is at most $k-1$.
By \cref{lem:chordal:repr-from-graphic} if suffices to find a max $r$-representative family for $\scal$ of the requested size.
This can be done with \cref{lem:prelim:repr:efficient-final}
in time $\Oh(2^{\omega k}\cdot k^{\omega})$.
Since $2 < 2^\omega$ the latter term is dominating in the running time.
\end{proof}

We also provide an efficient algorithm for processing large families of graphs that arise when handling a join node.
For a graph family $\gcal$ we write $\gcal \cap \texttt{chordal}$ to indicate the subfamily of chordal graphs in $\gcal$.

\begin{lemma}
\label{lem:tw:repr:product}
Consider two families of chordal graphs
$\mathcal{G}_1, \mathcal{G}_2 \sub \mathcal{G}_{X,B}$ for some pair $(X,B)$.
Suppose that $|X| = k > 0$, $|\mathcal G_1|, |\mathcal G_2| \le 2^{k}$, and each $G \in \mathcal{G}_1 \cup  \mathcal{G}_2$ has $\Oh(k)$ vertices.
Let $\mathcal{G} = \{ (G_1,X) \oplus (G_2,X) \mid G_1 \in \mathcal{G}_1, G_2 \in \mathcal{G}_2\} \cap \texttt{chordal}$
and $w: \mathcal{G} \to \mathbb{N}$ be a non-negative weight function.
Then $\widehat {\mathcal{G}} \sub_{\mathrm{maxrep}} \mathcal{G}$
of size at most $2^{k-1}$ can be computed in time $\Oh(2^{(\omega - 1)k}3^k \cdot k^{\omega})$ when given $\gcal_1, \gcal_2, \gcal, w$.
\end{lemma}
\begin{proof}
Let $\scal = \{\sign(G,X) \mid G \in \mathcal{G}\}$,
$\scal_1 = \{\sign(G,X) \mid G \in \mathcal{G}_1\}$,
and $\scal_2 = \{\sign(G,X) \mid G \in \mathcal{G}_2\}$.
We claim that $\scal = \scal_1 \bullet \scal_2$.
For $\ell \in \{1,2\}$ consider $S_\ell \in \scal_\ell$ and $G_\ell \in \gcal_\ell$ such that $\sign(G_\ell,X) = S_\ell$.
Then $(G_1,X) \oplus (G_2,X)$ belongs to $\gcal$ if and only if it is chordal what, by \cref{lem:chordal:criterion}, is equivalent to the condition that $S_1 \cap S_2 = \emptyset$ and $S_1 \cup S_2$ is independent in $M_B$.
This justifies the claim.

By \cref{lem:chordal:rank} the rank $r$ of $M_B$ is at most $k-1$.
We proceed similarly as in \cref{lem:tw:repr:efficient}
but we use \cref{lem:prelim:repr:product-final} to compute an $r$-representative family of size $2^{k-1}$ for $\scal_1 \bullet \scal_2$ in time $\Oh(2^{(\omega - 1)k}3^k \cdot k^{\omega})$.
Computing the signatures and the mapping $\tau \colon \scal \to \gcal$ takes time $|\gcal| \cdot k^{\Oh(1)} = 4^k \cdot k^{\Oh(1)}$
so the previous term is dominating in the running time.

\end{proof}

\subsection{Dynamic programming}

We present a dynamic programming algorithm processing a tree decomposition.
We begin with describing the states of the dynamic programming routine and their invariants.
Although we do not store tables indexed by partial solutions but rather a family of partial solutions with weights (as in~\cite{representative-efficient}), we keep the notion of `state' to refer to information stored at a node of a decomposition.

\subparagraph{States for DP.}
A state for a node $t \in V(\ttw)$ 
is a family of pairs $(\bcal_{t,X}, h_{t,X})$ assigned to each $X \sub \chi(t)$, where $\bcal_{t,X} \sub \gcal_{X,G[X]}$ is a family of chordal graphs and $h_{t,X} \colon \bcal_{t,X} \to \mathbb{N}$ is  a non-negative weight function.

Recall that  $V_t$ denotes the set of vertices occurring in the subtree rooted at $t \in V(\ttw)$ and $U_t = V_t \sm \chi(t)$.
The intended meaning of $H \sub \bcal_{t,X}$ and $h_{t,X}(H) = s$ is that there should exists a set $A \sub U_t$ of total weight $s$ so that $(A,X)$ is a partial solution equivalent to $H$.
For each $H$ we want to keep track of such a set $A$ of maximal weight.

For sets $X, Y$  we write concisely $(A,B) \sub (X,Y)$ to denote $A \sub X$, $B \sub Y$. 

\subparagraph{Correctness invariant.}
For $t \in V(\ttw)$ and $X \subseteq \chi(t)$
we say that a pair $( \bcal,  h)$ satisfies the correctness invariant for $(t,X)$ if the following holds.
\begin{enumerate}
    \item For each $H \in \bcal$ there exists a set $A \subseteq U_t$ so that $\compress(G[A \cup X], X) = H$, $G[A \cup X]$ is chordal, and $ h(H) = w(A)$. \label{item:correct1}
    \item For each $(A,B) \sub  (U_t, V(G) \sm V_t)$ for which $G[A \cup X \cup B]$ is chordal there exists $H \in \bcal$ so that $(H,X) \oplus (G[B\cup X],X)$ is chordal and $h(H) \ge w(A)$.\label{item:correct2} 
\end{enumerate}

As a consequence of this invariant, we have that for each pair $(A,B)$ from the second condition there exists $\widehat A \subseteq  U_t$ (a replacement for $A$) so that $w(\widehat A) \ge w(A)$, $\compress(G[\widehat A \cup X], X) \in \bcal$, and  $(\compress(G[\widehat A \cup X], X), X) \oplus (G[B \cup X], X])$ is chordal which implies that $G[\widehat A \cup X \cup B]$ is chordal (\cref{lem:chordal:condense-glue}).

Since every graph in $\bcal_{t,X}$ is condensed with respect to $X$,
\cref{lem:chordal:condense-size} implies a bound on its size (we do not subtract one from $2|X|$ to cover the case $X = \emptyset$).
 
\begin{observation}\label{lem:tw:consequence}
If $t \in V(\ttw)$, $X \subseteq \chi(t)$, and $( \bcal,  h)$ satisfies the correctness invariant for $(t,X)$, then every graph in $\bcal$ has at most $2|X|$ vertices.
\end{observation} 
 
We take advantage of the theory developed so far to keep the sizes of $\bcal_{t,X}$ in check.

\begin{lemma}\label{lem:tw:shrinking-correctness}
Consider $t \in V(\ttw)$ and $X \sub \chi(t)$.
If a pair $(\bcal, h)$ satisfies the correctness invariant for $(t,X)$ and $\widehat \bcal \sub_{\mathrm{maxrep}} \bcal$ then $(\widehat \bcal, h)$ also satisfies the correctness invariant for $(t,X)$.
\end{lemma}
\begin{proof}
The first condition is satisfied trivially as the quantification switches to a subset of $\bcal$.
Consider $(A,B) \sub (U_t, V(G) \sm V_t))$ from condition (2).
By the assumption, there exists $H \in \bcal$ so that $(H,X) \oplus (G[B \cup X],X)$ is chordal and $h(H) \ge w(A)$.
By the definition of a max-representative family there exists $\widehat H \sub \widehat \bcal$ so that $(\widehat H,X) \oplus (G[B \cup X],X)$ is chordal and $h(\widehat H) \ge h(H) \ge w(A)$.
The claim follows.
\end{proof}

\subparagraph{Size invariant.}
For $t \in V(\ttw)$ and $X \subseteq \chi(t)$
we say that $\bcal \sub \gcal_{X,G[X]}$ satisfies the size invariant if $|\bcal| \le 2^{|X|}$. 

For a node $t \in \ttw$ we say that its state satisfies the correctness or size invariant if all the pairs $(\bcal_{t,X}, h_{t,X})_{X \sub \chi(t)}$ satisfies it.

We move on to describing the dynamic programming routine for the three non-trivial types of nodes in a nice tree decomposition.
In each case we begin with the algorithm, then prove the correctness invariant, and then analyze the running time together with the size invariant.

\begin{lemma}[Introduce node]
\label{lem:tw:introduce}
Let $(\ttw, \chi)$ be a nice tree decomposition of $G$ of width $k$ and $t \in V(\ttw)$ be an introduce node with a child $t'$.
Suppose that the state for $t'$ satisfying the correctness and size invariants is given.
Then we can compute the state for $t$ which satisfies the correctness and size invariants in time $3^k k^{\Oh(1)}$.
\end{lemma}
\begin{proof}
We have $\chi(t) = \chi(t') \cup \{v\}$ for some vertex $v$.
Next, $U_t = U_{t'}$ and $V_t = V_{t'} \cup \{v\}$.
Note that $N_G(v) \cap U_t = \emptyset$.

We define operation {\texttt{Introduce}}$(H,N)$ which takes a graph $H$ and a set $N \sub V(H)$
and inserts to $H$ a new vertex $v$ with neighborhood $N$.

If $X \subseteq \chi(t)$ does not contain $v$, we 
set $\bcal_{t,X} = \bcal_{t',X}$ and $h_{t,X} = h_{t',X}$.
Consider $X \subseteq \chi(t)$ containing $v$; let $X^- = X \sm v$.
If $G[X]$ is not chordal we set $\bcal_{t,X} = \emptyset$.
Otherwise for each $H' \in \bcal_{t',X^-}$ we
compute $(H,X) = \texttt{Introduce}(H',N_G(v) \cap X^-)$.
If $H$ is chordal we insert it to  $\bcal_{t,X}$ and set $h_{t,X}(H) =  h_{t',X^-}(H')$.

\subparagraph{Correctness.}
Suppose first that $v \not\in X$.
We check condition (1).
Let $H \in \bcal_{t,X}$.
By the construction, $H \in \bcal_{t',X}$ and, by the inductive assumption, there is a set $A \subseteq U_{t'} = U_t$ so that $\compress(G[A \cup X], X) = H$, $G[A \cup X]$ is chordal, and $h_{t,X}(H) = h_{t',X}(H) = w(A)$.
To see condition (2), consider any
 $A \subseteq U_t$ and $B \subseteq V(G) \sm V_t$.
Then $A \subseteq U_{t'}$ and $B \subseteq V(G) \sm V_{t'}$ so again the claim follows directly from the invariant for $t'$.

Suppose now that $v \in X$.
We check condition (1).
Every $H \in \bcal_{t,X}$ is of the form $H = \texttt{Introduce}(H',N_G(v) \cap X^-)$ for some $H' \in \bcal_{t',X^-}$.
Moreover, $H$ must be chordal by the construction.
By the inductive assumption there is a set $A \sub U_{t'} = U_t$ so that $\compress(G[A \cup X^-],X^-) = H'$, $G[A \cup X^-]$ is chordal, $h_{t,X}(H) = h_{t',X}(H') = w(A)$.
Since $N_G(v) \cap A = \emptyset$ 
the vertex $v$ does not affect which connected components of $G[A]$ are simplicial in $G[A \cup X]$.
Therefore, the graph $\compress(G[A \cup X],X)$ can be obtained from $\compress(G[A \cup X^-],X^-)$ by simply inserting the vertex $v$, and so it equals $H$.

Finally, we need to check that $G[A \cup X]$ is chordal.
We apply criterion from \cref{lem:chordal:criterion-old} with respect to set $X^-$.
For every connected component $C$ of  $G[A \cup X] - X^-$ the graph $G[C \cup X^-]$ is either a subgraph of $G[A \cup X^-]$ or it equals $G[X]$---in both cases it is chordal.
The graph 
$\compress(G[A \cup X], X^-)$ is either isomorphic to $H$ or can be obtained from $H$ by removal of $v$ (when $v$ is simplicial in $H$).
Therefore this graph is also chordal what implies that $G[A \cup X]$ is chordal.

We move on to condition (2) for the case $v \in X$.
Let $(A, B) \sub (U_t, V(G) \sm V_t)$ be such that $G[A \cup X \cup B]$ is chordal.
Note that no such pair exists when $G[X]$ is not chordal, so we only care about the case where $G[X]$ is chordal.
Let $B^+ = B \cup \{v\}$;
then $(A, B^+) \sub (U_{t'}, V(G) \sm V_{t'})$. 
Therefore, there exists $H' \in \bcal_{t',X^-}$ so that $(H',X^-) \oplus (G[B^+ \cup X^-],X^-)$ is chordal and $h_{t',X^-}(H') \ge w(A)$.
Let $H = \texttt{Introduce}(H',N_G(v) \cap X^-)$.
By construction we have $h_{t,X}(H) = h_{t',X^-}(H')$.
Next, $B \cup X = B^+ \cup X^-$.
Since and $N_H(v) \sub X$, the gluing product $(H,X) \oplus (G[B \cup X],X)$ is the same as $(H',X^-) \oplus (G[B \cup X],X^-)$ which is chordal, as noted above.
This also implies that $H$ is chordal so it gets inserted to $\bcal_{t,X}$.
This concludes the proof of the correctness invariant. 

\subparagraph{Running time.}
If $v \not\in X$ then $\bcal_{t,X} = \bcal_{t',X}$ and
if $v \in X$ then each graph from $\bcal_{t',X^-}$ is mapped to a single graph in $\bcal_{t,X}$.
In both cases the size invariant is preserved.
For each $X \sub \chi(t)$ we process at most $2^{|X|}$ graphs 
what in total gives $\sum_{X \sub \chi(t)} 2^{|X|} \le 3^{k+1}$ graphs.
By \cref{lem:tw:consequence} each graph has at most $2(k+1)$ vertices and is processed in time $k^{\Oh(1)}$.
\end{proof}

Before describing the routine for a forget node we prove that the condensing operation applied with respect to $X$ and then to $X \sm v$ results in the same graph as when applying it directly to~$X \sm v$.

\begin{lemma}\label{lem:tw:compress-twice}
Let $X \sub V(H)$ and $v \in X$.
Next, let $H_2 = \emph\compress(H, X)$ and $H_3 = \emph\compress(H_2, X \sm v)$.
Then $H_3 = \emph\compress(H, X \sm v)$.
\end{lemma}
\begin{proof}
Let $C$ be a connected component of $H-X$ which is non-adjacent to $v$.
Then it gets contracted into a single vertex and is present in both $H_2, H_3$ as long as it is non-simplicial.
Consider now the connected components of $H-X$ which are adjacent to $v$.
Let $S_1, S_2, \dots, S_\ell$ be those among them which are simplicial in $H$ and $C_1, C_2, \dots, C_r$ are those which are non-simplicial.
Let $V_{S,C,v} = \bigcup S_i \cup \bigcup D_i \cup \{v\}$ and  $V_{C,v} = \bigcup C_i  \cup \{v\}$.
If $u \in N_H(S_i)$ then, since $v \in N_H(S_i)$ and $S_i$ is simplicial, we have $u \in N_H(v)$.
This implies that $N_H(V_{S,C,v}) = N_H(V_{C,v})$.
In $H_2$ all the components $S_i$ are removed and $C_i$ are contracted to vertices.
In $H_3$ the latter vertices are replaced with a new vertex $v'$ with neighborhood $N_H(V_{C,v})$.
In $\compress(H, X \sm v)$ we directly replace all components $S_i, C_i$ with a new vertex $v''$ with neighborhood $N_H(V_{S,C,v})$.
As observed above, these neighborhoods coincide, hence $H_3 = \compress(H, X \sm v)$.
\end{proof}

\begin{lemma}[Forget node]
\label{lem:tw:forget}
Let $(\ttw, \chi)$ be a nice tree decomposition of $G$ of width $k$ and $t \in V(\ttw)$ be a forget node with a child $t'$.
Suppose that the state for $t'$ satisfying the correctness and size invariants is given.
Then we can compute the state for $t$ which satisfies the correctness and size invariants in time $(2^\omega + 1)^k k^{\Oh(1)}$.
\end{lemma}
\begin{proof}
We have $\chi(t') = \chi(t) \cup \{v\}$ for some vertex $v \in V_t$.
Next, $V_{t} = V_{t'}$ and $U_{t} = U_{t'} \cup \{v\}$.
Note that $N_G(v) \sub V_t$.

Let $X \sub \chi(t)$, $X^+ = X \cup \{v\}$, and $\bcal' = \{\compress(H',X) \mid H' \in \bcal_{t',X^+}\}$.
For each $H \in \bcal_{t',X} \cup \bcal'$
we define its weight $h_{t,X}(H)$ as maximum over $h_{t',X}(H)$ (considered only if $H \in \bcal_{t',X}$) and $\max \{h_{t',X^+}(H') + w(v)) \mid H' \in \bcal_{t',X^+} \land \compress(H',X) = H\}$.
It follows from the construction that we take maximum over a non-empty set.
Finally, we compute $\bcal_{t,X}$ with \cref{lem:tw:repr:efficient} as the max-representative family for $\bcal_{t',X} \cup \bcal'$ with the weight function $h_{t,X}$.

\subparagraph{Correctness.}
We show that the pair $(\bcal_{t',X} \cup \bcal', h_{t,X})$ satisfies the correctness invariant.
Then the claim for $(\bcal_{t, X}, h_{t,X})$ will follow from \cref{lem:tw:shrinking-correctness}.

We check condition (1).
First consider $H \in \bcal_{t',X}$ satisfying $h_{t,X}(H) = h_{t',X}(H)$.
By the inductive assumption there is a set $A \subseteq U_{t'} \subset U_t$ so that $\compress(G[A \cup X], X) = H$, $G[A \cup X]$ is chordal, and $h_{t,X}(H) = w(A)$, as intended.

Now consider $H \in \bcal'$ for which there exists $H' \in \bcal_{t',X^+}$ so that $H = \compress(H',X)$ and $h_{t,X}(H) = h_{t',X^+}(H') + w(v)$.
We know that there exists a set $A \subseteq U_{t'} = U_t \sm v$ so that $\compress(G[A \cup X^+], X^+) = H'$, $G[A \cup X^+]$ is chordal, and $h_{t',X^+}(H') = w(A)$.
The set $A^+ = A \cup \{v\} \sub U_t$ satisfies  $h_{t,X}(H) = w(A^+)$ and $G[A^+ \cup X] = G[A \cup X^+]$ is chordal.
From \cref{lem:tw:compress-twice} we obtain that $\compress(G[A^+ \cup X], X) = \compress(H',X) = H$.

We move on to condition (2).
Let $(A,B) \sub (U_t, V(G) \sm V_t)$ be that $G[A \cup X \cup B]$ is chordal.
First consider the case $v \not\in A$.
Then $(A,B) \sub (U_{t'}, V(G) \sm V_{t'})$ and there exists $H \in \bcal_{t',X}$ so that $(H,X) \oplus (G[B \cup X], X)$ is chordal and $h_{t',X}(H) \ge w(A)$.
By construction $h_{t,X}(H) \ge h_{t',X}(H)$.

Now suppose that $v \in A$ and let $A^- = A \sm v$.
We have $(A^-,B) \sub (U_{t'}, V(G) \sm V_{t'})$ and so there exists $H' \in \bcal_{t',X^+}$ so that $(H',X^+) \oplus (G[B \cup X^+], X^+)$ is chordal and $h_{t',X^+}(H') \ge w(A^-) = w(A) - w(v)$.
Let $H = \compress(H',X) \in \bcal'$.
Since $N_G(v) \cap B = \emptyset$, we deduce that $(H',X^+) \oplus (G[B \cup X^+], X^+) = (H',X) \oplus (G[B \cup X], X)$.
The graph $(H,X) \oplus (G[B \cup X], X)$ can be obtained from the one above by a series edge contractions and possibly a vertex removal so it is chordal as well.
Finally, $h_{t,X}(H) \ge h_{t',X^+}(H') + w(v) \ge w(A)$.

\subparagraph{Running time.}
Consider a non-empty $X \sub \chi(t)$.
By \cref{lem:tw:consequence} each graph $H \in \bcal_{t',X} \cup \bcal'$ has at most $2|X|$ vertices and is processed in time $k^{\Oh(1)}$.
By the assumption $|\bcal_{t',X}|\le 2^{|X|}$ and
$|\bcal_{t',X^+}|\le 2^{|X|+1}$ so the input to \cref{lem:tw:repr:efficient} has size less than $2^{|X|+2}$.
The computation of the max-representative family $\bcal_{t,X}$ takes time $\Oh(2^{\omega \cdot |X|}\cdot |X|^{\omega})$.
This family have size at most $2^{|X|-1}$ so it satisfies the size invariant.
The sum of the exponential terms in the total running time equals $\sum_{X \sub \chi(t)} 2^{\omega \cdot |X|} = (2^{\omega} +1)^{|X|}$.
\end{proof}

\begin{lemma}[Join node]
\label{lem:tw:join}
Let $(\ttw, \chi)$ be a nice tree decomposition of $G$ of width $k$ and $t \in V(\ttw)$ be a join node with a children $t_1, t_2$.
Suppose that the states for $t_1,t_2$ satisfying the correctness and size invariants are given.
Then we can compute the state for $t$ which satisfies the correctness and size invariants in time $\Oh\br{(2^{\omega-1}\cdot 3 + 1)^k \cdot k^{\omega}}$.
\end{lemma}
\begin{proof}
We have $\chi(t) = \chi(t_1) = \chi(t_2)$ and $U_t = U_{t_1} \cup U_{t_2}$ where the union is disjoint.

Consider $X \subseteq \chi(t)$.
Let $\bcal' = \{(H_1,X) \oplus (H_2,X) \mid H_1 \in \bcal_{t_1,X}, H_1 \in \bcal_{t_1,X}\} \cap \texttt{chordal}$. 
For each $H \in \bcal'$
we define its weight $h_{t,X}(H)$ as $\max \left(w_{t_1,X}(H_1) + w_{t_2,X}(H_2)\right)$ over $\{(H_1,X) \oplus (H_2,X) = H \mid H_1 \in \bcal_{t_1,X}, H_1 \in \bcal_{t_1,X}\}$.
We compute $\bcal_{t, X}$ with \cref{lem:tw:repr:product} as the max-representative family for $\bcal'$ with respect to the weight function $h_{t,X}$.

\subparagraph{Correctness.} We show the correctness invariant for the pair $(\bcal', h_{t,X})$. Then the claim will follow from \cref{lem:tw:shrinking-correctness}. 
Let us fix $X \subseteq \chi(t)$.

We check condition (1).
Let $H \in \bcal'$ and $H_1 \in \bcal_{t_1,X}, H_2 \in \bcal_{t_2,X}$ be such that $H = (H_1,X) \oplus (H_2,X)$ and $h_{t,X}(H) = w_{t_1,X}(H_1) + w_{t_2,X}(H_2)$.
By the assumption, there exist sets $A_1 \sub U_{t_1}, A_2 \sub U_{t_2}$ so that for $i \in \{1,2\}$ it holds that $\compress(G[A_i \cup X], X) = H_i$, $G[A_i \cup X]$ is chordal, and $w_{t_i,X}(H_i) = w(A_i)$.
Observe that there are no edges between $A_1, A_2$
so $G[A_1 \cup A_2 \cup X] = (G[A_1 \cup X], X) \oplus (G[A_2 \cup X], X)$.
By \cref{lem:chordal:condense-assocciate} we have $\compress(G[A_1 \cup A_2 \cup X], X) = \compress(G[A_1 \cup X], X) \oplus \compress(G[A_2 \cup X], X) = (H_1,X) \oplus (H_2,X) = H$.
Since $H$ is chordal by the construction, \cref{lem:chordal:criterion-old} implies that $G[A_1 \cup A_2 \cup X]$ is chordal.
It remains to check that $w(A_1 \cup A_2) = w(A_1) + w(A_2) = h_{t,X}(H)$.

Condition (2).
Let $(A,B) \sub (U_t, V(G) \sm V_t)$ be that $G[A \cup X \cup B]$ is chordal, and $A_1 = A \cap U_{t_1}, A_2 = A \cap U_{t_2}$.
Note that $B \cup A_2 \sub V(G) \sm V_{t_1}$ and there are no edges between $U_{t_1}$ and $U_{t_2}$.
By the assumption, there exists $H_1 \in \bcal_{t_1,X}$ so that $(H_1,X) \oplus (G[B \cup A_2 \cup X)$ is chordal and $h_{t_1,X}(H_1) \ge w(A_1)$.
From condition (1) we obtain that there exists $\widehat A_1 \subseteq U_{t_1}$ so that $\compress(G[\widehat A_1 \cup X], X) = H_1$, $G[\widehat A_1 \cup X]$ is chordal,  and $h_{t_1,X}(H_1) = w(\widehat A_1)$.
It follows from \cref{lem:chordal:condense-glue} that $G[\widehat A_1 \cup A_2 \cup X \cup B]$ is chordal.
Next, we consider $(A_2, B \cup \widehat A_1) \sub (U_{t_2}, V(G) \sm V_{t_2})$.
There exists $H_2 \in \bcal_{t_2,X}$ so that $(H_2,X) \oplus (G[B \cup \widehat A_1 \cup X),X)$ is chordal and $h_{t_2,X}(H_2) \ge w(A_2)$.
Again from condition (1) we obtain  $\widehat A_2 \subseteq U_{t_2}$ so that $\compress(G[\widehat A_2 \cup X], X) = H_2$, $G[\widehat A_2 \cup X]$ is chordal, and $h_{t_2,X}(H_2) = w(\widehat A_2)$.
As before, the graph $G[\widehat A_1 \cup \widehat A_2 \cup X \cup B]$ is chordal.
By \cref{lem:chordal:condense-assocciate} we have that $H = (H_1, X) \oplus (H_2,X)$ equals
$\compress(G[\widehat A_1 \cup \widehat A_2 \cup X], X)$.
Then $(H,X) \oplus (G[B \cup X], X)$ is chordal
which in particular means that $H$ is chordal and belongs to $\bcal'$.
Finally, we check that $h_{t,X}(H) \ge h_{t_1,X}(H_1) + h_{t_2,X}(H_2) \ge w(A_1) + w(A_2) = w(A)$.

\subparagraph{Running time.}
Consider a non-empty $X \sub \chi(t)$.
By \cref{lem:tw:consequence} each graph $H \in \bcal'$ has at most $2|X|$ vertices.
By the assumption $|\bcal_{t_1,X}|, |\bcal_{t_2,X}| \le 2^{|X|}$ so we can use \cref{lem:tw:repr:product} to compute
the max-representative family $\bcal_{t,X}$ in time $\Oh \br{(2^{\omega-1} \cdot 3)^{|X|}\cdot |X|^{\omega}}$.
This family have size at most $2^{|X|-1}$ so it satisfies the size invariant.
The sum of the exponential terms in the total running time equals $\sum_{X \sub \chi(t)} (2^{\omega-1} \cdot 3)^{|X|} = (2^{\omega-1} \cdot 3 + 1)^{|X|}$.
\end{proof}

\thmChordal*
\begin{proof}
Let $(\ttw, \chi)$ be a tree decomposition of $G$ of width $k$.
We can assume that this is a nice tree decomposition and $|V(\ttw)| = \Oh(nk)$.
We fill the states for $t \in V(\ttw)$ in a standard bottom-up fashion, while maintaining the correctness and size invariants.
When $t$ is a base node, we have $\chi(t) = V_t =  \emptyset$ and it suffices to consider $X = \emptyset$.
The family $\bcal_{t,\emptyset}$ then contains only an empty graph with weight zero.
This satisfies both invariants trivially.
We proceed the remaining types of nodes using Lemmas~\ref{lem:tw:introduce}, \ref{lem:tw:forget}, and \ref{lem:tw:join}.
The bottleneck for the running time comes from processing a join node. 

After filling the state of the root node $r$,
we read the value of $h_{r,\emptyset}(\bot)$, where $\bot$ denotes the empty graph.
We claim that this value equals the highest weight of a vertex set in $G$ inducing a chordal graph.
We have $\chi(r) = \emptyset$, $V_t = U_t = V(G)$.
The family $\bcal_{r,\emptyset}$ can contain only the empty graph.
Let $A \sub V(G)$ be a maximal-weight set inducing a chordal graph.
From the correctness condition (2) 
we obtain that there exists $H \in \bcal_{r,\emptyset}$ (clearly $H = \bot$) so that $h_{r,\emptyset}(H) \ge w(A)$.
From the correctness condition (1) 
for $H = \bot$ we get that 
there exists $\widehat A \sub V(G)$ so that $G[\widehat A \cup \emptyset]$ is chordal and $h_{r,\emptyset}(\bot) = w(\widehat A)$.
This implies that $h_{r,\emptyset}(\bot) = w(A)$.
\end{proof}
}

\section{Interval Deletion}
\label{sec:interval}

We switch our attention to \textsc{Interval Vertex Deletion} and show that in this case it is unlikely to achieve any speed-up over the existing $2^{\Oh(\tw \log \tw)}\cdot n$-time algorithm.
We prove \cref{thm:intro:interval} via a~parameterized reduction from $k\times k$ \textsc{Permutation Clique}, which is defined as follows.

\defproblem{$k\times k$ {Permutation Clique}}
{Graph $G$ over the vertex set $[k] \times [k]$.}
{Is there a permutation $\pi \colon [k] \to [k]$ so that $(1,\pi(1)), (2,\pi(2)), \dots, (k,\pi(k))$ forms a~clique in $G$?}

\subparagraph{Permutation gadget.}
We will encode a permutation $\pi \colon [k] \to [k]$ as a family of sets $N_1, N_2, \dots, N_k$ so that $N_i = \pi([i])$ (i.e., $N_i$ is the set of $i$ numbers appearing first in $\pi$). 
First, we need a gadget to verify that such a family represents some permutation.

\begin{definition}
For an integer $k$, let $Y_k$ be a graph on a vertex set
$\{y_1, y_2, \dots, y_{k+2}\}$
 so that $\{y_1, y_2, \dots, y_{k+1}\}$ induces a clique and $y_{k+2}$ is adjacent only to $y_{k+1}$.
\end{definition}

\short{\refstepcounter{theorem}}
\full{
We need a simple observation that every linearly ordered family of sets can be represented by some permutation.

\begin{lemma}
\label{lem:interval:permutation-helper}
Let $N_1, \dots, N_\ell \subseteq [k]$.
Suppose that for each $i,j \in [\ell]$ it holds that $N_i \subseteq N_j$ or $N_j \subseteq N_i$.
Then there exists a permutation $\pi \colon [k] \to [k]$ so that for each $i \in [\ell]$ it holds that $N_i = \pi([n_i])$ where $n_i = |N_i|$.
\end{lemma}
\begin{proof}
Proof by induction on $k$.
For $k = 1$ the claim clearly holds so consider $k > 1$. 
If a set $N_i$ is empty then $N_i = \pi(\emptyset)$ for 
any permutation, so we can assume that all the sets are non-empty.
The family $(N_i)_{i\in[\ell]}$ is linearly ordered so the after reordering the indices we can assume that $N_1 \subseteq N_2 \subseteq \dots \subseteq N_\ell$.
Let $e$ be an arbitrary element from $N_1$ and $\tau \colon [k] \setminus \{e\} \to [k-1]$ be an arbitrary bijection.
We define $N'_i = \tau(N_i \setminus \{e\})$.
Then $N'_1 \subseteq N'_2 \subseteq \dots \subseteq N'_\ell$.
By the inductive assumption, there exists a permutation 
$\pi' \colon [k-1] \to [k-1]$ such that $N'_i = \pi'([n_i - 1])$.
We define $\pi(1) = e$ and for $i>1$ as follows: $\pi(i) = \tau^{-1}(\pi'(i-1))$.
Then $N_i = \{e\} \cup \tau^{-1}(N'_i) = \{\pi(1)\} \cup \tau^{-1}(\pi'([n_i - 1])) = \{\pi(1)\} \cup \pi([2, n_i]) = \pi([n_i])$.
\end{proof}
}

\begin{figure}
    \centering
    \includegraphics[width=0.7\linewidth]{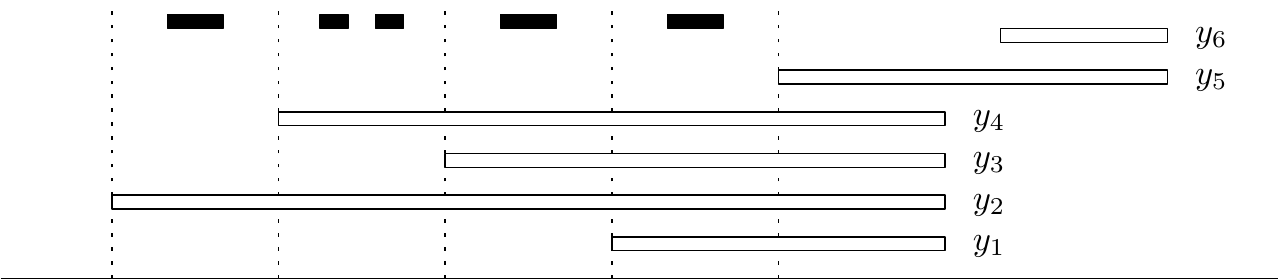}
    \caption{Illustration for \cref{lem:interval:permutation}.
    The intervals for vertices of $Y_4$ are blank, ordered from bottom to top.
    They encode permutation $(2,4,3,1)$.
    The black intervals represent vertices $x_1,x_2,x_3,x_4,x_5$ with neighborhoods encoding sets $\{2\}$, $\{2,4\}$ (twice), $\{2,4,3\}$, and $\{2,4,3,1\}$.
    }
    \label{fig:permutation}
\end{figure}

We shall enforce a linear order on  $N_1, \dots, N_k$ by demanding that a particular supergraph of $Y_k$ is interval.
The corresponding interval model is depicted on \Cref{fig:permutation}.

\begin{lemma}[\starr]
\label{lem:interval:permutation}
Let $N_1, \dots, N_\ell \subseteq [k]$.
Consider a graph $G$ obtained from $Y_k$ by inserting an independent set of vertices $x_1, \dots, x_\ell$ so that $N_G(x_i) = \{y_j \mid j \in N_i\}$.
Then $G$ is interval if and only if there exists a permutation $\pi \colon [k] \to [k]$ so that for each $i \in \ell$ it holds that $N_i = \pi([n_i])$ where $n_i = |N_i|$.
\end{lemma}
\full{
\begin{proof}
Suppose there is no such permutation.
By \cref{lem:interval:permutation-helper}
there are $i,j \in [\ell]$ so that neither $N_i \subseteq N_j$ nor $N_j \subseteq N_i$.
Fix $p_i \in N_G(x_i) \setminus N_G(x_j)$, and $p_j \in N_G(x_j) \setminus N_G(x_i)$.
We claim that $(x_i, x_j, y_{k+2})$ forms an AT in $G$.
Indeed, $(x_i, p_i, y_{k+1}, y_{k+2})$ avoids $N_G[x_j]$,
$(x_j, p_j, y_{k+1}, y_{k+2})$ avoids $N_G[x_i]$,
and $(x_i, p_i, p_j, x_j)$ avoids $N_G[y_{k+2}] = \{y_{k+1},y_{k+2}\}$.
Since $G$ contains an AT, it is not interval.

Now suppose that a permutation $\pi$ satisfying the conditions of the lemma exists.
We construct an interval model of $G$ (see \Cref{fig:permutation}).
Let $\eps = \frac 1 {3\ell}$.
A vertex $y_i \in V(Y_k)$ where $i \in [k]$
is assigned the interval $[\pi^{-1}(i), k+2]$.
Vertex $y_{k+1}$ is assigned the interval $[k+1,k+4]$ and vertex $y_{k+2}$ is assigned to $[k+3, k+4]$.
For $i \in [\ell]$ we consider the vertex $x_i$ with neighborhood specified by $N_i = \pi([n_i])$.
Its interval is given as $[n_i + \frac{i-1}{\ell} + \eps, n_i + \frac{i-1}{\ell} + 2\eps]$.
Note that the intervals of distinct $x_i, x_j$ are disjoint, as intended.
The interval of $x_i$ is contained in $(n_i, n_i+1)$ so it intersects an interval of the form $[\pi^{-1}(j), k+2]$ exactly when $\pi^{-1}(j) \le n_i$.
The latter inequality is equivalent to $j \in \pi([n_i]) = N_i$.
The lemma follows.
\end{proof}
}

\begin{figure}[t]
    \centering
    \includegraphics[width=0.75\linewidth]{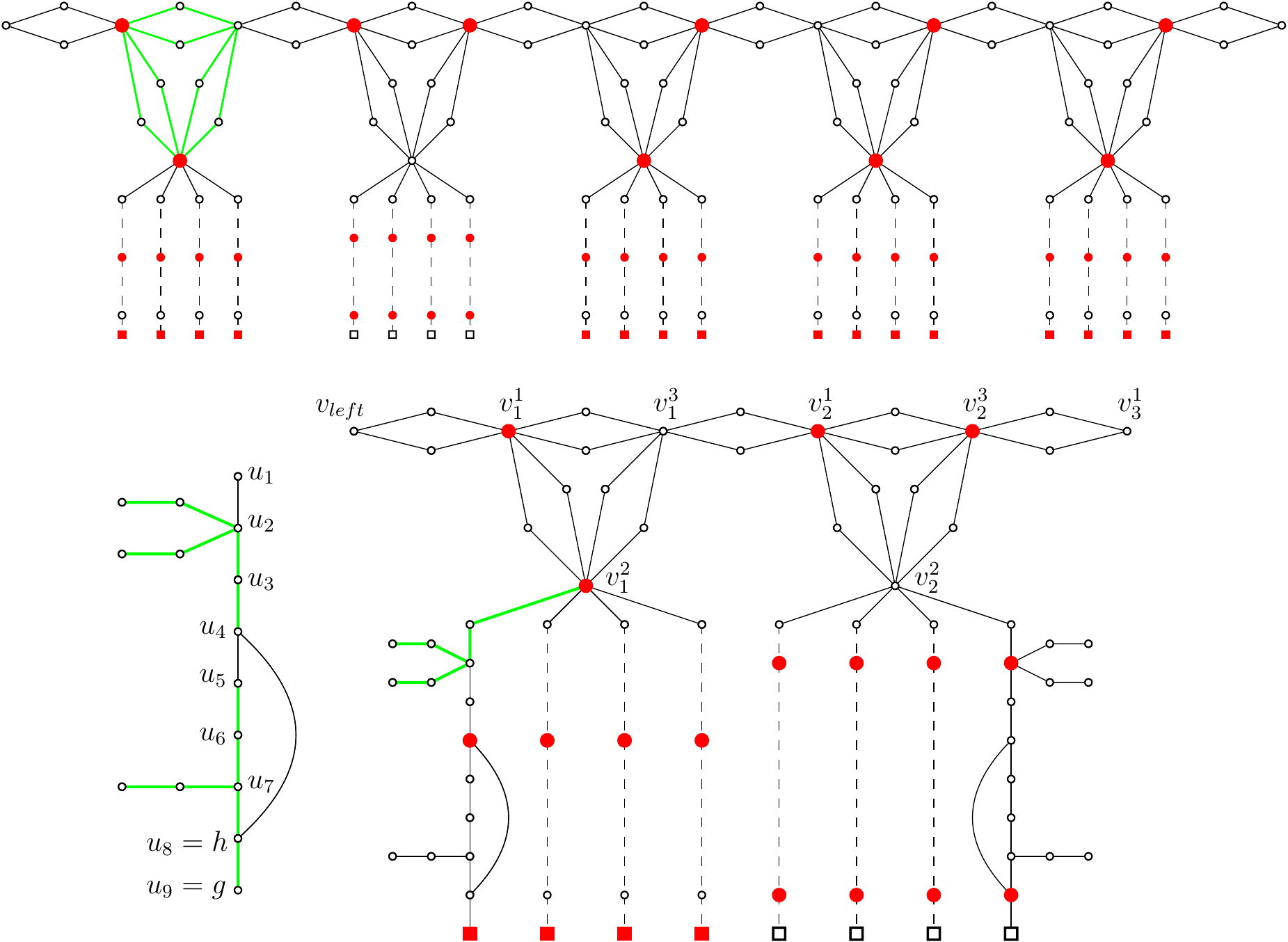}
    \caption{Top: the choice gadget $H_5$ with the subgraph $Q_1$ highlighted in green.
    The copies of $P$ are sketched symbolically with dashed lines and the squares represent vertices $g^\alpha_i$.
    The red disks and squares represent a solution constructed in \cref{lem:interval:choice-properties}(2).
    This solution `chooses' $i=2$, leaves untouched the four vertices  $g^\alpha_2$, and removes  $h^\alpha_2$ as well as $g^\alpha_i$ for $i \ne 2$.
    Bottom left: the graph $P$ and vertices named $h,g$. Two vertex-disjoint non-interval subgraphs of $P$  
    have green edges.
    Bottom right: a closer look at the first two blocks of $H_5$ with two copies of $P$ drawn in detail. 
    The subgraph highlighted in green witnesses that if a minimum-size solution removes $g^\alpha_i$ for at least one $\alpha \in [4]$ then it must also remove $v^2_i$, what is exploited in \cref{lem:interval:choice-properties}(3).}
    \label{fig:selector}
\end{figure}

\subparagraph{Choice gadget.}
We need to verify that $(i,\pi(i))(j,\pi(j)) \in E(G)$ for each $1 \le i < j \le k$.
As $\pi(i)$ is the only element in $N_i \sm N_{i-1}$, the information whether $(i,\pi(i)),(j,\pi(j)) \in E(G)$ can be extracted from the tuple $(N_{i-1},N_i,N_{j-1},N_j)$.
We construct a gadget that enforces a solution to select one such valid tuple. 

We use a following convention to describe the gadgets.
When $P$ is a graph with a distinguished vertex named $v$
and a graph $H$ is constructed using explicit vertex-disjoint copies of the graph $P$, referred to as $P_1, P_2, \dots, P_\ell$, 
we refer to the copy of $v$ within the subgraph $P_i$ as $P_i[v]$.
We construct the choice gadget as a path-like structure consisting of blocks, each equipped with four special vertices.
These are the only vertices that later get connected to the permutation gadget.
On the intuitive level, a solution should choose one block, leave its special vertices untouched, and remove the remaining special vertices.
See \Cref{fig:selector} for an illustration.

\begin{definition}
The graph $P$ is obtained from a path $(u_1, u_2, \dots, u_9)$ by appending to $u_2$ two subdivided edges, one subdivided edge to $u_7$, and inserting edge $u_4u_8$.

The choice gadget of order $s$ is a graph constructed as follows.
We begin with a vertex set $\bigcup_{i=1}^s \{v^1_i, v^2_i, v^3_i\} \cup \{v_{\text{left}}, v_{\text{right}}\}$.
For each pair $(x,y)$ of the form $(v^1_i, v^2_i), (v^2_i, v^3_i), (v^3_i, v^1_i), (v^3_i, v^1_{i+1})$ as well as for $(v_{\text{left}}, v^1_1)$, $(v^3_s, v_{\text{right}})$ we create two subdivided edges between $x$ and $y$.
We refer to the subgraph given by 
the two subdivided edges between $x,y$ as $\langle x, y \rangle$.
We refer to the union of $\langle v^1_i, v^2_i \rangle, \langle v^2_i, v^3_i \rangle, \langle v^3_i, v^1_i \rangle$ as $Q_i$.

Next, for each $i \in [s]$ we create four copies of the graph $P$, denoted $P^1_i, P^2_i, P^3_i, P^4_i$.
We insert edges between $v^2_i$ and 
$P^1_i[u_1], P^2_i[u_1], P^3_i[u_1], P^4_i[u_1]$.
We refer to vertices 
$P^\alpha_i[u_8]$, $P^\alpha_i[u_9]$, $\alpha \in [4]$, as respectively 
$h^\alpha_i$, $g^\alpha_i$.
\end{definition}

 The choice gadget is designed to enforce a special structure of minimum-size interval deletion sets.
\short{We exploit the fact that $P$ contains two vertex-disjoint subgraphs with asteroidal triples (see \Cref{fig:selector}) so any interval deletion set in a choice gadget must contain at least two vertices from each copy of $P$.
\refstepcounter{theorem} }

\full{
\begin{lemma}
\label{lem:interval:choice-lb}
Let $H_s$ be the choice gadget of order $s$ and $X$ be an interval deletion set in $H_s$.
Then for each $i \in [s]$ and $\alpha \in [4]$ it holds that $|V(P^\alpha_i) \cap X| \ge 2$ and $|V(Q_i) \cap X| \ge 2$.
\end{lemma}
\begin{proof}
The graph $P$ contains two vertex-disjoint non-interval subgraphs which is witnessed by ATs: one induced by $u_2,u_3,u_4$ and the two subdivided edges appended to $u_2$, the second one induced by $u_5,u_6,u_7,u_8,u_9$ and the subdivided edge appended to $u_7$ (see \Cref{fig:selector}).
Therefore any copy of $P$ in $H_s$ must contain at least two vertices from $X$.
Next, observe that no single vertex intersects all three holes in $Q_i$.
Therefore any interval deletion set must contain at least two vertices from $V(Q_i)$.
\end{proof}
}

We prove several properties of the choice gadget
which are analogous to the properties of the gadget used by Pilipczuk in the lower bound for \textsc{Planar Vertex Deletion}~\cite{Pilipczuk17}.
However, in that construction every block has only one special vertex with edges leaving the gadget, while in our case there are four special vertices.
We also need to ensure that when the special vertices in some block are not being removed then a solution can remove their neighbors in the gadget.
(Inserting a planar graph attached to a single vertex of $G$ does not affect planarity of $G$ but the analogous property does not hold for the class of interval graphs.)
The special structure of the graph $P$ allows us to resolve these two issues.


\begin{restatable}[\starr]{lemma}{lemIntChoice}
\label{lem:interval:choice-properties}
Let $H_s$ be the choice gadget of order $s$.
\begin{enumerate}[nolistsep]
    \item The minimal size of an interval deletion set in $H_s$ is $10 s$.
    \item For every $i \in [s]$ there exists a minimum-size interval deletion set $X$ in $H_s$ such that $\{h^1_i,h^2_i,h^3_i,h^4_i\} \sub X$ and $\{g^1_j,g^2_j,g^3_j,g^4_j\} \sub X$ for each $j \ne i$.
    \item For every minimum-size interval deletion set $X$ in $H_s$ there is $i \in [s]$ such that \\ \mbox{$\{g^1_i,g^2_i,g^3_i,g^4_i\} \cap X = \emptyset$.}
    \item If $s \le 2^k$ then $\td(H_s) \le \td(H_1) + k$, where $\td(G)$ stands for the treedepth of $G$. 
\end{enumerate}
\end{restatable}
\full{
\begin{proof}
Part (1). All the  $4s$ copies of $P$, as well as subgraphs $Q_1, \dots, Q_s$, are vertex-disjoint in $H_s$.
The lower bound follows from \cref{lem:interval:choice-lb} whereas the upper bound is a consequence of the next part of the lemma. 

Part (2). The construction  is depicted on \Cref{fig:selector}.
The set $X$ comprises:
\begin{itemize}[nolistsep, itemindent = 0.2in]
    \item  $v^1_j, v^2_j$ for $j <i$,
    \item  $v^1_i, v^3_i$,
    \item  $v^2_j, v^3_j$ for $j > i$,
    \item  $P^\alpha_j[u_4], P^\alpha_j[u_9]$ for $\alpha \in [4]$, $j \ne i$, 
    \item $P^\alpha_i[u_2], P^\alpha_i[u_8]$ for $\alpha \in [4]$.
\end{itemize}
One can easily verify that $|X| = 10s$ and each connected component of $H_s - X$ is either a path or a star with at most two subdivided edges. 
These graphs are interval.

Part (3).
Suppose that there exists an interval deletion set $X$ of size  $10s$ such that for each  $i \in [s]$
there is $\alpha \in [4]$ so that $g^\alpha_i \in X$.
From \cref{lem:interval:choice-lb} we know that
$|X \cap V(P^\alpha_i)| \ge 2$.
From a counting argument we infer that in fact it must be $|X \cap V(P^\alpha_i)| = 2$.
The vertices  $P^\alpha_i[u_4], P^\alpha_i[u_5], P^\alpha_i[u_6], P^\alpha_i[u_7], P^\alpha_i[u_8]$ induce $C_5$ so one of them must belong to $X$.
Together with $g^\alpha_i = P^\alpha_i[u_9]$ these are the two vertices of $X \cap V(P^\alpha_i)$.

The vertices  $v^2_i, P^\alpha_i[u_1], P^\alpha_i[u_2]$ and the two subdivided edges appended to $P^\alpha_i[u_2]$
induce a graph with an AT (see \Cref{fig:selector}).
As no more vertices from $V(P^\alpha_i)$ belong to $X$ apart from the two described above, it must be $v^2_i \in X$.

This argument works for every   $i \in [s]$.
We count the already allocated vertices:
\[
\left| \bigcup_{i \in [s], \alpha \in [4]} (X \cap V(P_i^\alpha))\, \right| + | \{v^2_i \mid i \in [s]\}| = 8s + s = 9s.
\]

Since $|X| = 10s$, there are exactly $s$ vertices remaining in $X$.
But there are $s+1$ vertex-disjoint holes yet to be hit: $\langle v_{\text{left}}, v^1_1 \rangle, \langle v^3_1, v^1_2 \rangle, \langle v^3_2, v^1_3 \rangle, \dots, \langle v^3_s, v_{\text{right}} \rangle$.
This means that $X$ cannot be an interval deletion set in $H_s$.

Part (4). 
Clearly $\td(H_s) \le \td(H_{s+1})$ so it suffices to prove the claim for $s = 2^k$ by induction on $k$.
For $k=0$ we get equality.
For $k > 0$ there exists a vertex $v \in V(H_{2^k})$ so that $H_{2^k} - v$ has two connected components, each being a subgraph of $H_{2^{k-1}}$.
By the definition of treedepth we get $\td(H_{2^{k}}) \le \td(H_{2^{k-1}}) + 1$.
\end{proof} }

Lokshtanov et al.~\cite{LokshtanovMS18} proved that $k\times k$ \textsc{Permutation Clique} cannot be solved in time $2^{o(k \log k)}$ assuming ETH.
According to the reduction below, this also rules out running time of the form $2^{o(\td \log \td)} \cdot n^{\Oh(1)}$ for \textsc{Interval Vertex Deletion}, where $\td$ is the treedepth of the input graph.
As $\tw(G) \le \td(G)$, this entails the same hardness for treewidth, what proves \cref{thm:intro:interval}.

\begin{proposition}
There is an algorithm that, given an instance $(G,k)$ of $k\times k$ \textsc{Permutation Clique}, runs in time $2^{\Oh(k)}$ and returns an equivalent unweighted instance $(H,p)$ of \textsc{Interval Vertex Deletion} such that $|V(H)| = 2^{\Oh(k)}$ and $\td(H) = \Oh(k)$.
\end{proposition}
\begin{proof}
For $1 \le i < j \le k$ and $x \ne y \in [k]$ let $\scal_{i,x,j,y}$ be the family of tuples $(S_1,S_2,S_3,S_4)$ of subsets of $[k]$ satisfying:
\begin{itemize}[nolistsep]
    \item $S_1 \subset S_2 \sub S_3 \subset S_4$,
    \item $|S_1| = i-1$,
    \item $S_2 \sm S_1 = \{x\}$,
    \item $|S_3| = j-1$,
    \item $S_4 \sm S_3 = \{y\}$.
\end{itemize}
Furthermore, for $1 \le i < j \le k$, let $\scal_{i,j}$
be the union of $\scal_{i,x,j,y}$ over all pairs $x \ne y \in [k]$ such that $(i,x)(j,y) \in E(G)$.
Let $s_{i,j} = |\scal_{i,j}|$ and $\rho_{i,j} \colon [s_{i,j}] \to \scal_{i,j}$ be an arbitrary bijection. Clearly $s_{i,j} \le 4^k k^2$.

The graph $H$ consists of a permutation gadget $Y_k$
and, for each $1 \le i < j \le k$, a choice gadget $C_{i,j}$ of order $s_{i,j}$.
For $S \subseteq [k]$ we use shorthand $Y_k[S] = \{y_i \mid i \in S\}$.
For $\ell \in [s_{i,j}]$ and $(S_1,S_2,S_3,S_4) = \rho_{i,j}(\ell)$ the vertices $C_{i,j}[g^1_\ell]$, $C_{i,j}[g^2_\ell]$, $C_{i,j}[g^3_\ell]$, $C_{i,j}[g^4_\ell]$
get connected to vertex sets $Y_k[S_1], Y_k[S_2], Y_k[S_3], Y_k[S_4]$, respectively. 
This finishes the construction of $H$.
The number of vertices in $H$ is clearly $2^{\Oh(k)}$ and the construction can be performed in time polynomial in the size of $H$. 
We set $p = 10 \cdot \sum_{1\le i < j \le k} s_{i,j}$.

\begin{claim}
If $(G,k)$ admits a solution, then $H$ has an interval deletion set of size $p$.
\end{claim}
\begin{innerproof}
Let $\pi \colon [k] \to [k]$ be a permutation encoding a clique in $G$.
By the construction, for each $1 \le i < j \le k$ we have $\left(\pi([i-1]), \pi([i]), \pi([j-1]), \pi([j]\right) \in \scal_{i,j}$.
Let $\ell \in [s_{i,j}]$ be the index mapped to this tuple by $\rho_{i,j}$.
By \cref{lem:interval:choice-properties}(2) the choice gadget $C_{i,j}$ has an interval deletion set $X_{i,j} \subseteq V(C_{i,j})$ of size $10 s_{i,j}$ such that 
 $\{C_{i,j}[h^1_\ell],C_{i,j}[h^2_\ell],C_{i,j}[h^3_\ell],C_{i,j}[h^4_\ell]\} \sub X_{i,j}$ 
 and $\{C_{i,j}[g^1_r],C_{i,j}[g^2_r],C_{i,j}[g^3_r],C_{i,j}[g^4_r]\} \sub X_{i,j}$ for each $r \ne \ell$.
In other words, $X_{i,j}$ contains all vertices in $C_{i,j}$ which are adjacent to $Y_k$ except for the $C_{i,j}$-copies of $g^1_\ell,g^2_\ell,g^3_\ell,g^4_\ell$ and $X_{i,j}$ also contains the neighbors of  
$C_{i,j}[g^1_\ell], C_{i,j}[g^2_\ell], C_{i,j}[g^3_\ell], C_{i,j}[g^4_\ell]$~in~$C_{i,j}$.

We set $X = \bigcup_{1 \le i < j \le k} X_{i,j}$.
Then the only connected component of $H-X$ which is not a connected component of any $C_{i,j} - X_{i,j}$ is given by $Y_k$ together with an independent set of the vertices described above.
The neighborhood of each such vertex in $Y_k$ is of the form $Y_k[\pi([k'])]$ for some $0 \le k' \le k$.
By \cref{lem:interval:permutation} this component is an interval graph.
This shows that $X$ is indeed an interval deletion set.
\end{innerproof}

\begin{claim}
If $H$ has an interval deletion set of size at most $p$, then $(G,k)$ admits a solution.
\end{claim}
\begin{innerproof}
Let $X$ be an interval deletion set in $H$.
By \cref{lem:interval:choice-properties}(1) a minimum-size interval deletion set in $C_{i,j}$ has size $10 s_{i,j}$.
As the choice gadgets are vertex-disjoint subgraphs of $H$, the set $X$ must contain exactly $10 s_{i,j}$ vertices from $V(C_{i,j})$.
This also implies that $V(Y_k) \cap X = \emptyset$.

Let $X_{i,j} = V(C_{i,j}) \cap X$.
By \cref{lem:interval:choice-properties}(3) there exists $\ell \in [s_{i,j}]$ such that $\big\{C_{i,j}[g^1_\ell], C_{i,j}[g^2_\ell]$, $C_{i,j}[g^3_\ell], C_{i,j}[g^4_\ell] \,\big\} \cap X_{i,j} = \emptyset$.
Therefore for each pair $(i,j)$ there is a tuple $(S^1_{i,j}, S^2_{i,j}, S^3_{i,j}, S^4_{i,j}) \in \scal_{i,j}$ so that vertices from $C_{i,j}$ with neighborhoods $Y_k[S^1_{i,j}], Y_k[S^2_{i,j}], Y_k[S^3_{i,j}], Y_k[S^4_{i,j}]$ are present in $H-X$.
By \cref{lem:interval:permutation} there exists a single permutation  $\pi \colon [k] \to [k]$ so that each set $S^\alpha_{i,j}$ is of the form $\pi([|S^\alpha_{i,j}|])$.
By the definition of family $\scal_{i,j}$ this implies that $(i,\pi(i))(j,\pi(j)) \in E(G)$ for each pair $(i,j)$.
Hence there is a $k$-clique in $G$.
\end{innerproof}

\begin{claim}
The treedepth of $H$ is $\Oh(k)$.
\end{claim}
\begin{innerproof}
The treedepth of $H$ is at most $|Y_k| = k+2$ plus $\td(H - Y_k)$, which equals the maximum
of $\td(C_{i,j})$ over all employed 
choice gadgets $C_{i,j}$.
As $s_{i,j} \le 4^kk^2$, \cref{lem:interval:choice-properties}(4) implies that $\td(C_{i,j}) \le 2k + 2\log_2 k + \Oh(1)$.
\end{innerproof}
This conludes the proof of the proposition.
\end{proof}

\full{
\subsection{Upper bound}
\label{app:intervalDP}

Saitoh et al.~\cite{intervalDP} have presented an algorithm for \textsc{Interval Edge Deletion} (and for several related graphs classes) with running time $2^{\Oh(\tw \log \tw)}\cdot n$ and stated that they expect it to also work for the vertex-deletion variant~\cite[\S 6]{intervalDP}.
We briefly describe their approach and justify that vertex deletion can indeed be incorporated.

Instead of working with real-line interval models, one can represent an interval model in an abstract way.
For a set $X$ its \emph{interval representation} is a linear order $\pi$ over the set $LR_X = L_X \cup R_X \cup \{\bot, \top\}$ where $L_X = \{\ell_x \mid x \in X\}$,  $R_X = \{r_x \mid x \in X\}$, such that $\bot <_\pi \ell_x <_\pi r_x <_\pi \top$ for each $x \in X$.
An interval is a pair of elements from $LR_X$.
The interval graph $G_\pi$ of an interval representation $\pi$ over $V$ is defined by adding an edge $uv$ whenever $(\ell_u, r_u)$ and $(\ell_v, r_v)$ intersect in $\pi$.
It is clear that a graph $G$ is interval if and only if
there exists an interval representation $\pi$ over $V$ such that $G = G_\pi$.

Let $(\ttw, \chi)$ be a rooted tree decomposition of $G$.
The state of $t \in V(\ttw)$ is a set of triples $(\pi, I, c)$ where $\pi$ is an interval representation over $\chi(t)$ such that $G_\pi$ is a subgraph of $G[\chi(t)]$, $I$ is a set of intervals from $LR_{\chi(t)}$, and $c \in \mathbb{N}$.
For each interval representation $\tau$ over $V_t$, for which $G_\tau$ is a subgraph of $G[V_t]$ we define its abstraction as a triple $(\pi, I, c)$ where $\pi$ is a restriction of $\tau$ to $\chi(t)$, for each connected component $C$ of $G[U_t]$ there is an interval $(\ell_c, r_c) \in I$ so that $(\ell_c, r_c)$ is a inclusion-wise minimal interval from $LR_\chi(t)$ that contains all the intervals from $C$, and $c = |E(G[V_t]) - E(G_\tau) - E(G[\chi(t)])|$ counts the number of edges from $G[V_t]$ with at most one endpoint in $\chi(t)$ which are not present in $G_\tau$.
We write $I \sqsubseteq_\pi I'$ if every interval from $I$ is contained in some interval from $I'$.
We say that $(\pi,I,c)$ dominates $(\pi,I',c')$ if $I \sqsubseteq_\pi I'$ and $c \le c'$.

Saitoh et al. show that when $(\pi,I,c)$ dominates $(\pi,I',c')$ there is no need to store $(\pi,I',c')$  in the state for $t$.
To see this, consider an interval representation $\tau'$ over $V(G)$ which corresponds to some valid solution, so that $(\pi,I',c')$ is an abstraction of $\tau'$ restricted to $V_t$.
Since there are no edges between $U_t$ and $V(G) \sm V_t$, for every $v \in V(G) \sm V_t$ and connected component $C$ of $G[U_t]$ the interval of $v$ must be disjoint the interval spanned by $C$ with respect to $\tau'$.
Because $I \sqsubseteq_\pi I'$,
we can modify $\tau'$ to obtain a new interval representation $\tau$ over $V(G)$ which coincides on $U_t$ with some partial solution whose abstraction is $(\pi,I,c)$.
As $c \le c'$ the number of edges deleted with respect to $\tau$ does not grow compared to $\tau'$. 
The last observation is that when $|\chi(t)| = k$ and no triple stored at $t$ dominates another triple then their number is $2^{\Oh(k \log k)}$.

In order to adapt the algorithm for vertex deletion, we can store tuples $(X,\pi,I,c)$ where $X \sub\chi(t)$ is the set of vertices which are not deleted, $\pi$ is an interval representation over $X$ so that $G_\pi = G[X]$, $I$ is a set of intervals from $LR_X$, and $c \in \mathbb{N}$.
Now a partial solution is a triple $(A,X,\tau)$ where $A \sub U_t, X \sub \chi(t)$, and $\tau$ is an interval representation over $A \cup X$ so that $G_\tau = G[A \cup X]$.
The abstraction of $(A,X,\tau)$ is $(X,\pi,I,c)$ where
$\pi$ is a restriction of $\tau$ to $X$, for each connected component $C$ of $G[A]$ there is an interval $(\ell_c, r_c) \in I$ so that $(\ell_c, r_c)$ is a inclusion-wise minimal interval from $LR_X$ that contains all the intervals from $C$, and $c = |U_t| - |A|$ counts the number of vertices deleted so far.
As before, we say that $(X,\pi,I,c)$ dominates $(X,\pi,I',c')$ if $I \sqsubseteq_\pi I'$ and $c \le c'$.
By the same argument as before, we can neglect the tuples which are dominated and bound the number of tuples stored for a pair $(t,X)$ by $2^{\Oh(k \log k)}$.
Since there are $2^k$ choices for $X$, the general upper bound follows.
It is easy to see that both algorithms can be also extended to incorporate weights.
}

\section{Conclusion and open problems}
\label{sec:conclusion}

We have obtained ETH-tight bounds for vertex-deletion problems into the classes of chordal and interval graphs, under the treewidth parameterization.
The status of the corresponding edge-deletion problems remains unclear (see~\cite{intervalDP}).
The related problem, \textsc{Feedback Vertex Set}, can be solved using representative families within the same running time as our algorithm for \textsc{ChVD}~\cite{representative-product}.
However, it admits a faster deterministic algorithm based on the determinant approach~\cite{wlodarczyk2019clifford} and an even faster randomized algorithm based on the Cut \& Count technique~\cite{cygan2011solving}.
Could \textsc{ChVD} also be amenable to one of those techniques?

Our algorithm for \textsc{ChVD} is based on a novel
connection between chordal graphs and graphic matroids, which might come in useful in other settings.
In particular, we ask whether this insight can be leveraged to improve the running time for \textsc{ChVD} parameterized by the solution size $k$, where
the current-best algorithm runs in time $2^{\Oh(k \log k)}n^{\Oh(1)}$~\cite{CaoM16}.
A~direct avenue for a potential improvement would be to reduce the problem in time $2^{\Oh(k)}n^{\Oh(1)}$ to the case with treewidth $\Oh(k)$ and then apply \cref{thm:intro:chordal}.
Such a strategy has been employed in the state-of-the-art algorithm for \textsc{Planar Vertex Deletion} parameterized by the solution~size~\cite{JansenLS14}.

\bibliography{bib}
\end{document}